\documentclass[11pt]{article}
	\usepackage{amsfonts, amsmath, amssymb, amsthm}
	\usepackage{amsthm}
	\usepackage{bbm}
	\usepackage{amssymb}
	\usepackage{thmtools}
	\usepackage{algpseudocode}
	\usepackage{graphicx}
	\usepackage{stmaryrd}
	\usepackage[show]{ed}
	\usepackage[small]{complexity}
	\usepackage{ulem}
	\usepackage{xspace}
	\usepackage[margin=1in]{geometry}
	\usepackage[pdftex,pagebackref=true]{hyperref}
	\usepackage{algorithm}

	\numberwithin{equation}{section}

	\theoremstyle{plain}
	\declaretheorem[numberlike=equation]{theorem}
	\declaretheorem[unnumbered,name=Theorem]{theorem*}
	\declaretheorem[numberlike=equation]{lemma}
	\declaretheorem[unnumbered,name=Lemma]{lemma*}
	\declaretheorem[numberlike=equation]{corollary}
	\declaretheorem[unnumbered,name=Corollary]{corollary*}
	
	\declaretheorem[unnumbered,name=Proposition]{proposition*}
	
	\declaretheorem[unnumbered,name=Claim]{claim*}
	
	\declaretheorem[unnumbered,name=Conjecture]{conjecture*}
	\declaretheorem[numberlike=equation]{construction}
	\declaretheorem[unnumbered,name=Construction]{construction*}

	\declaretheorem[numberlike=equation]{definition}
	\declaretheorem[unnumbered,name=Definition]{definition*}
	
	\declaretheorem[unnumbered,name=Example]{example*}
	
	\declaretheorem[unnumbered,name=Notation]{notation*}

	\theoremstyle{remark}
	\declaretheorem[numberlike=equation]{remark}
	\declaretheorem[unnumbered,name=Remark]{remark*}

	\newcommand{\zr}[1]{{\llb {#1}\rrb}}
	\newcommand{\ind}[1]{\mathbbm{1}_{#1}}
	\newcommand{\bits}{\{0,1\}}

	\newcommand{\llb}{\llbracket}
	\newcommand{\rrb}{\rrbracket}

	\DeclareMathOperator{\rank}{rank}
	\DeclareMathOperator{\sgn}{sgn}
	
	\DeclareMathOperator{\spn}{span}
	\DeclareMathOperator{\rspn}{row-span}

	\DeclareMathOperator{\coeffo}{\mathfrak{C}}
	\DeclareMathOperator{\trexp}{\lfloor\exp\rfloor}
	\DeclareMathOperator{\evaldim}{evaldim}

	\newcommand{\coeff}[1]{\coeffo_{#1}}
	\newcommand{\homp}[1]{H_{#1}}

	\newcommand{\F}{\mathbb{F}}
	\renewcommand{\K}{\mathbb{K}}
	
	\newcommand{\Z}{\mathbb{Z}}
	\newcommand{\N}{\mathbb{N}}
	\newcommand{\Q}{\mathbb{Q}}

	\newcommand{\cA}{\mathcal{A}}
	\newcommand{\cC}{\mathcal{C}}
	\newcommand{\cG}{\mathcal{G}}
	\newcommand{\cH}{\mathcal{H}}
	\newcommand{\cM}{\mathcal{M}}
	\newcommand{\cN}{\mathcal{N}}
	\newcommand{\cR}{\mathcal{R}}
	\newcommand{\Id}{\mathrm{I}}
	\newcommand{\nulls}{\lambda}

	\newcommand{\eqdef}{:=}
	\renewcommand{\O}{\mathcal{O}}

	\mathchardef\mhyphen="2D
	\DeclareMathOperator{\sdeg}{s-deg}

	\title{Quasipolynomial-time Identity Testing of Non-Commutative and Read-Once Oblivious Algebraic Branching Programs}
	\author{%
	Michael A.\ Forbes\thanks{Email: \texttt{miforbes@mit.edu}, Department of Electrical Engineering and Computer Science, MIT CSAIL, 32 Vassar St., Cambridge, MA 02139,
	Supported by NSF Grant CCF-0939370.}
			\and
	Amir Shpilka\thanks{Faculty of Computer Science, Technion --- Israel Institute of Technology, Haifa, Israel, \texttt{shpilka@cs.technion.ac.il}.  The research leading to these results has received funding from the European Community's Seventh Framework Programme (FP7/2007-2013) under grant agreement number 257575.}}
	\date{\today}

\begin{document}

\maketitle

\begin{abstract}
	We study the problem of obtaining efficient, deterministic, \textit{black-box polynomial identity testing algorithms (PIT)} for algebraic branching programs (ABPs) that are read-once and oblivious.  This class has an efficient, deterministic, white-box polynomial identity testing algorithm (due to Raz and Shpilka~\cite{RazShpilka05}), but prior to this work there was no known such black-box algorithm. 

	The main result of this work gives the first quasi-polynomial sized hitting sets for size $S$ circuits from this class, when the order of the variables is known. As our hitting set is of size $\exp(\lg^2 S)$, this is analogous  (in the terminology of boolean pseudorandomness) to a seed-length of $\lg^2 S$, which is the seed length of the pseudorandom generators of Nisan~\cite{Nisan92} and Impagliazzo-Nisan-Wigderson~\cite{INW94} for read-once oblivious \emph{boolean} branching programs. Thus our work can be seen as an algebraic analogue of these foundational results in boolean pseudorandomness.
	
	Our results are stronger for branching programs of bounded width, where we give a hitting set of size $\exp(\lg^2S/\lg\lg S)$, corresponding to a seed length of $\lg^2S/\lg\lg S$. This is in stark contrast to the known results for read-once oblivious boolean branching programs of bounded width, where no pseudorandom generator (or hitting set) with seed length $o(\lg^2 S)$ is known.  Thus, while our work is in some sense an algebraic analogue of existing boolean results, the two regimes seem to have non-trivial differences.
	
	In follow up work (\cite{ForbesShpilka13}), we strengthened a result of Mulmuley~\cite{Mulmuley12}, and showed that derandomizing a particular case of the Noether Normalization Lemma is reducible to black-box PIT of read-once oblivious ABPs.  Using the results of the present work, this gives a derandomization of Noether Normalization in that case, which Mulmuley conjectured would difficult due to its relations to problems in algebraic geometry.

	We also show that several other circuit classes can be black-box reduced to read-once oblivious ABPs, including set-multilinear ABPs (a generalization of depth-3 set-multilinear formulas), non-commutative ABPs (generalizing non-commutative formulas), and (semi-)diagonal depth-4 circuits (as introduced by Saxena~\cite{Saxena08}).  For set-multilinear ABPs and non-commutative ABPs, we give quasi-polynomial-time black-box PIT algorithms, where the latter case involves evaluations over the algebra of $(D+1)\times (D+1)$ matrices, where $D$ is the depth of the ABP.  For (semi-)diagonal depth-4 circuits, we obtain a black-box PIT algorithm (over any characteristic) whose run-time is quasi-polynomial in the runtime of Saxena's white-box algorithm, matching the concurrent work of Agrawal, Saha, and Saxena~\cite{AgrawalSS12}.  Finally, by combining our results with the reconstruction algorithm of Klivans and Shpilka~\cite{KlivansShpilka06}, we obtain deterministic reconstruction algorithms for the above circuit classes.
\end{abstract}

\thispagestyle{empty}
\newpage
\pagenumbering{arabic}

\section{Introduction}

Let $C$ be an algebraic circuit in the input variables $x_1,\ldots,x_n$, over a field $\F$. The output $C(x_1,\ldots,x_n)$ is a polynomial $f$ in the ring $\F[x_1,\ldots,x_n]$.  The polynomial identity testing (PIT) problem is to efficiently determine ``$f\equiv 0$?''.  In particular, we are asking if the formal expression $f$, as a polynomial in $\F[x_1,\ldots,x_n]$, is zero. Schwartz and Zippel~\cite{Zippel79,Schwartz80} showed that if $0\ne f \in \F[x_1,\ldots,x_n]$ is a polynomial of degree $\le d$, and $\alpha_1,\ldots,\alpha_n \in S\subseteq\F$ are chosen uniformly at random, then $f(\alpha_1,\ldots,\alpha_n) =0$ with probability at most\footnote{Note that this is meaningful only if $d < |S| \leq |\F|$, which in particular implies that $f$ is not the zero function.} $\le d/|S|$. Thus, given the circuit $C$, we can perform these evaluations efficiently, giving an efficient randomized procedure for answering ``$f\equiv 0$?''.  An important open problem is to find a derandomization of this algorithm, that is, to find a {\em deterministic} procedure for PIT that runs in polynomial time (in the size of the circuit $C$).

One interesting property of the above randomized algorithm of Schwartz-Zippel is that the algorithm does not need to ``see'' the circuit $C$. Namely, the algorithm only uses the circuit to compute the evaluations $f(\alpha_1,\ldots,\alpha_n)$.  Such an algorithm is called a {\em black-box} algorithm. In contrast, an algorithm that can access the internal structure of the circuit $C$ is called a {\em white-box} algorithm.  Clearly, the designer of the algorithm has more resources in the white-box model and so one can expect that solving PIT in this model should be a simpler task than in the black-box model.

The problem of derandomizing PIT has received a lot of attention in the past few years. In particular, many works examine a particular class of circuits $\cC$, and design PIT algorithms only for circuits in that class. One reason for this attention is the strong connection between deterministic PIT algorithms for a class $\cC$ and lower bounds for $\cC$. This connection was first observed by Heintz and Schnorr~\cite{HeintzSchnorr80} (and also by Agrawal~\cite{Agrawal05}) for the black-box model and by Kabanets and Impagliazzo~\cite{KabanetsImpagliazzo04} for the white-box model (see also Dvir, Shpilka and Yehudayoff~\cite{DSY09}). Another motivation for studying the problem is its relation to algorithmic questions. Indeed, the famous deterministic primality testing algorithm of Agrawal, Kayal and Saxena~\cite{AKS04} is based on derandomizing a specific polynomial identity. Finally, the PIT problem is, in some sense, the most general problem that we know today for which we have randomized $\coRP$ algorithms but no polynomial time algorithms, thus studying it is a natural step towards a better understanding of the relation between $\RP$ and $\P$. For more on the PIT problem we refer to the survey by Shpilka and Yehudayoff~\cite{SY10}.

Although the white-box model seems to be simpler than the black-box model, for most models for which a white-box PIT algorithm is known, a black-box PIT algorithm is also known, albeit sometimes with worse parameters. Such examples include depth-2 circuits (also known as sparse polynomials) \cite{Ben-OrTiwari88,KlivansSpielman01}, depth-3 $\Sigma\Pi\Sigma(k)$ circuits \cite{SaxenaSeshadhri11}, Read-$k$ formulas \cite{AvMV11} and depth-3 tensors (also known as depth-3 set-multilinear circuits) \cite{RazShpilka05,ForbesShpilka12}.  While the running time of the algorithms for depth-2 circuits and $\Sigma\Pi\Sigma(k)$ circuits are essentially the same in the white-box and black-box models, for Read-$k$ formulas and set-multilinear depth-3 circuits we encounter some loss that results in a quasi-polynomial running time in the black-box model compared to a polynomial time algorithm in the white-box model.

Until this work, the only model for which an efficient white-box algorithm was known without a (sub-exponential) black-box counterpart was the model of {\em non-commutative} algebraic formulas, or, more generally, the models of non-commutative algebraic branching programs (ABPs) and set-multilinear algebraic branching programs \cite{RazShpilka05} (see \autoref{sec:the model} for definitions).

The main result in this paper is a quasi-polynomial time PIT algorithm in the black-box model for read-once oblivious algebraic branching programs. Equivalently, we give a \textit{hitting set} of size $2^{\O(\lg^2 S)}$ for size $S$ circuits from this model. By tuning parameters in our recursion, we obtain hitting sets of size $2^{\O(\lg^2 S/\lg\lg S)}$ when the branching programs are also of bounded width. Using our main result we obtain black-box algorithms of similar running times for the models of set-multilinear ABPs, non-commutative ABPs, as well as for diagonal circuits as defined by Saxena~\cite{Saxena08}. Although exponential lower bounds are known for these models, we note that the algebraic hardness-versus-randomness result of Kabanets and Impagliazzo~\cite{KabanetsImpagliazzo04} (as well as the extension of this result by Dvir, Shpilka and Yehudayoff~\cite{DSY09} to low-depth circuits) does not imply a black-box PIT algorithm for the model, since their technique does not work for the restricted models studied here.

The algebraic circuit models considered in this work, though restricted, have received significant attention in existing work on lower bounds and pseudorandomness, and are strong enough to capture non-trivial derandomization questions arising elsewhere. In \autoref{sec:the model} we define these models, and state our results in \autoref{sec:results}. In \autoref{sec:related} we discuss various work concerning these models, and explain relations to other areas such as (boolean) space-bounded derandomization, the Noether Normalization Lemma from algebraic geometry, and an algebraic analogue of the natural proofs barrier in boolean lower bounds.  In \autoref{sec:read-once oblivious:overview}, we outline the main proof ideas of the hitting set for read-once ABPs. 


\subsection{The model}\label{sec:the model}

The model that we consider in this paper is that of set-multilinear algebraic branching programs.  In fact, we will work with a slightly more general model, but we first describe the model of set-multilinear ABPs.

Algebraic branching programs were first defined in the work of Nisan~\cite{Nisan91} who proved exponential lower bounds on the size of non-commutative ABPs computing the determinant or permanent polynomials.

\begin{definition}[Nisan]\label{def: ABP}
	An \textbf{algebraic branching program (ABP)} is a directed acyclic graph with one vertex of in-degree zero, which is called the {\rm source}, and one vertex of out-degree zero, which is called the {\rm sink}. The vertices of the graph are partitioned into levels numbered $0,\ldots,D$. Edges may only go from level $i-1$ to level $i$ for $i=1,\ldots,D$. The source is the only vertex at level $0$ and the sink is the only vertex at level $D$. Each edge is labeled with a affine function in the input variables. The width of an ABP is the maximum number of nodes in any layer, and the size of the ABP is the number of vertices.
\end{definition}


Each directed source-sink path in the ABP computes a polynomial, which is a product of the labels on the edges in that path. As this work concerns non-commutative computation, we specify that the product of the labels is in the order of the path, from source to sink.  The ABP itself computes the sum of all such polynomials.

We consider a slight variation of this model which we call set-multilinear ABP, in line with the term coined by Nisan and Wigderson~\cite{NisanWigderson96}. In the set-multilinear scenario the variables are partitioned into sets $$X=X_1 \sqcup X_2 \sqcup \cdots \sqcup X_D,$$ where $$X_i = \{x_{i,1},\ldots,x_{i,n}\}.$$ A set-multilinear monomial is a monomial of the form $$m=x_{1,i_1}\cdot x_{2,i_2}\cdots x_{D,i_D}.$$ In words, a set-multilinear monomial is a multilinear monomial that contains exactly one variable from each $X_i$. A set-multilinear polynomial is a polynomial consisting of set-multilinear monomials. In other words, the coefficients of a set-multilinear polynomial can be viewed as map from $[n]^D$ to the field $\F$, and thus is an $D$-dimensional tensor.

\begin{definition}[Set-multilinear ABP]\label{def: set-mult ABP}
	\sloppy A \textbf{set-multilinear algebraic branching program (ABP)} in the variable set $X=X_1 \sqcup X_2 \sqcup \cdots \sqcup X_D$ is an ABP of depth $D$, such that each edge between layer $i-1$ and layer $i$ is labeled with a (homogeneous) linear function in the variable set $X_i$.
\end{definition}

It is clear from the definition that a set-multilinear ABP computes a set-multilinear polynomial.  It is also not hard to see that any set-multilinear polynomial can be computed by a set-multilinear ABP.

In fact, our result holds for the following model, that we call read-once oblivious ABPs, as well.

\begin{definition}[Read-Once Oblivious ABP]\label{def: read-once oblivious ABP}
	A \textbf{read-once oblivious ABP} in the variable set $X=\{x_1,\ldots,x_D\}$ is an ABP of depth $D$, such that each edge between layer $i-1$ and layer $i$ is labeled with a univariate polynomial in $x_i$ of degree $<n$.
\end{definition}

Note that unlike previous definitions, in read-once oblivious ABPs we allow edges to be labeled with arbitrary univariate polynomials (with a bound of $n$ on the degree) and not just with linear forms. Observe that the mapping $x_{i,j} \leftrightarrow x_i^j$ transforms any set-multilinear ABP into a read-once oblivious ABP and vice versa (when we index $j$ starting at zero).

\subsection{Our results}\label{sec:results}

A black-box PIT algorithm is also known as an explicit hitting set, which we now define. We phrase the definition in generality to capture the notion of hitting sets for non-commutative polynomials, generalizing the usual notion.

\begin{definition}[Hitting Set]
	Let $\cC$ be a class of non-commutative $n$-variate polynomials, with coefficients in $\F$.  Let $\cR$ be a non-commutative ring with a (commutative) ring homomorphism $\F\to \cR$, so that polynomials in $\cC$ are defined over $\cR$.  A set $\cH\subseteq \cR^n$ is a \textbf{hitting set for $\cC$} if for all $f\in\cC$, $f\equiv 0$ iff $f|_\cH\equiv 0$.

	The hitting set $\cH$ is \textbf{$t(n)$-explicit} if given an index into $\cH$, the corresponding element of $\cH$ can be computed in $t(n)$-time.
\end{definition}

When $\cC$ is commutative, we will always use $\cR=\F$, and when $\cC$ is non-commutative, we will take $\cR=\F^{m\times m}$ for some appropriate $m$.

Our main result is a quasi-polynomial time black-box PIT algorithm for read-once oblivious ABPs.

\begin{theorem*}[\autoref{thm:main}, PIT for read-once oblivious ABPs]
	Let $\cC$ be the set of $D$-variate polynomials computable by width $r$, depth $D$, individual degree $<n$ read-once oblivious ABPs.  If $|\F|\ge \poly(D,n,r)$, then $\mathcal{C}$ has a $\poly(D,n,r)$-explicit hitting set, of size $\le \poly(D,n,r)^{\O(\lg D)}$.
\end{theorem*}

This theorem plays a crucial role in future work by the authors (\cite{ForbesShpilka13}), were we give a derandomization of Noether's Normalization Lemma in a certain case.  In \autoref{sec:read-once oblivious:overview} we explain our proof technique and give an overview of the proof. Our technique also yields an improved derandomization when the branching program has small width, which we now state.

\begin{theorem*}[\autoref{thm:small width}, PIT for small width read-once oblivious ABPs]
	\sloppy	Let $\cC$ be the set of $D$-variate polynomials computable by width $r\le \O(1)$, depth $D$, individual degree $<n$ read-once oblivious ABPs.  If $|\F|\ge \poly(D,n)$, then $\mathcal{C}$ has a $\poly(D,n)$-explicit hitting set, of size $\le \poly(D,n)^{\O(\lg D/\lg\lg D)}$.
\end{theorem*}

Using \autoref{thm:main} we obtain black-box PIT algorithms for several related models. We first observe that PIT for set-multilinear ABPs is an immediate corollary of Theorem~\ref{thm:main}.  As with read-once oblivious ABPs, we assume that we know the partition of the variables into the $D$ sets, and our results do not hold under permutation of the variables.

\begin{theorem*}[\autoref{thm:main:sm}, PIT for set-multilinear ABPs]
	Let $X=X_1 \sqcup X_2 \sqcup \cdots \sqcup X_D,$ where $X_i = \{x_{i,1},\ldots,x_{i,n}\}$, be a known partition. Let $\mathcal{C}$ be the set of set-multilinear polynomials $f(X_1,\ldots,X_D):\F^{nD}\to\F$ computable by a width $r$, depth $D$, set-multilinear ABP.  If $|\F|\ge \poly(D,n,r)$, then $\mathcal{C}$ has a $\poly(D,n,r)$-explicit hitting set, of size $\le \poly(D,n,r)^{\O(\lg D)}$.
\end{theorem*}

Next, we consider the model of non-commutative ABPs (see \autoref{sec: ncABP}).
Raz and Shpilka~\cite{RazShpilka05} gave a polynomial time white-box PIT algorithm for this model and we obtain a quasi-polynomial time black-box PIT algorithm. The evaluation points in our hitting set are vectors of $(D+1)\times (D+1)$ matrices for ABPs of depth $D$. We explain this choice in \autoref{sec: ncABP}.  In contrast to the above two results, this result for non-commutative ABPs makes no assumption about the variable ordering, and thus still holds under permutation of the variables.


\begin{theorem*}[\autoref{cor:ncABP}, Black-box PIT for non-commutative ABPs]
	\sloppy Let $\mathcal{NC}$ be the set of $n$-variate non-commutative polynomials computable by width $r$, depth $D$ ABPs. If $|\F|\ge \poly(D,n,r)$, then $\mathcal{NC}$ has a $\poly(D,n,r)$-explicit hitting set over $(D+1)\times(D+1)$ matrices, of size $\le \poly(D,n,r)^{\O(\lg D)}$.
\end{theorem*}

Saxena~\cite{Saxena08} defined the model of {\em diagonal circuits} (see \autoref{sec:diagonal} for definition) in an attempt to capture some of the complexity of depth-4 circuits. Relying on the algorithm of \cite{RazShpilka05}, Saxena obtained a polynomial time white-box PIT algorithm for this model (for a certain setting of the parameters).  Saha, Saptharishi and Saxena~\cite{SahaSS11}, with the essentially same techniques, later extended Saxena's white-box results to the so-called semi-diagonal model. Simultaneously and independently of our work, Agrawal, Saha and Saxena~\cite{AgrawalSS12} gave a black-box PIT algorithm for semi-diagonal depth-4 circuits that runs in quasi-polynomial time in the runtime of the known white-box algorithm, and works over fields of large characteristic.

Using our main theorem and a new reduction from diagonal circuits to read-once oblivious ABPs (Lemma~\ref{lem: diagonal to ABP}) we obtain a PIT algorithm for diagonal circuits that runs in time quasi-polynomial in the white-box algorithm given by Saxena.  Further, we do so over any characteristic in a unified way.  As with our result on non-commutative ABPs, this result makes no assumptions about variable order.  With essentially no modification, we obtain similar results for semi-diagonal circuits.

\begin{theorem*}[\autoref{thm: semi diagonal}, Black-box PIT for semi-diagonal circuits]
	Let $\F$ be a field of arbitrary characteristic.  Let $\mathcal{SDC}$ be the set of $n$-variate polynomials computable by semi-diagonal depth-4 circuits, that is, of the form $\Phi=\sum_{i\in[k]}m_i(\vec{x})\cdot P_{i,1}^{e_{i,1}}\cdots P_{i,r_i}^{e_{i,r_i}}$, where $m_i(\vec{x})$ is a monomial of degree $\le d$, $\prod_{j=1}^{r_i}(1+e_{i,j})\le e$ and $P_{i,j}$ is a sum of univariate polynomials of degree $\le d$.  Then if $|\F|\ge \poly(n,k,d,e)$ then $\mathcal{SDC}$ has a $\poly(n,k,e,d)$-explicit hitting set of size $\le \poly(n,d,k,e)^{\O(\lg n)}$.
\end{theorem*}

Ramprasad Saptharishi \cite{Saptharishi12} showed  that a new notion called {\em evaluation dimension} extends the partial derivative technique as used in \cite{Saxena08} and that our hitting set holds for this model as well, thus obtaining an alternate proof of \autoref{thm: semi diagonal}. We discuss this in Section~\ref{sec: eval-dim}.\\

Klivans and Shpilka~\cite{KlivansShpilka06} gave a polynomial-time learning algorithm for read-once oblivious ABPs\footnote{The terminology used in \cite{KlivansShpilka06} is that of learning by multiplicity automata, but these are completely equivalent to read-once oblivious ABPs, as shown in \cite{KlivansShpilka06}.}, given membership and equivalence queries. That is, they give a deterministic algorithm that can learn an unknown $f$ computed be a small read-once oblivious ABP, given a oracle that evaluates $f$ (``membership query''), as well an oracle such that for any hypothesis $h$ computed by a small read-once oblivious ABP with $h\ne f$, the oracle returns an evaluation point $\vec{x}$ such that $h(\vec{x})\ne f(\vec{x})$ (``equivalence query''). Membership queries can be implemented deterministically given black-box access to $f$, and equivalence queries can be implemented using random queries.  That is, by appealing to the Schwartz-Zippel algorithm, distinguishing $h\ne f$ can be done using random evaluations.

Our above results construct hitting sets for read-once oblivious ABPs, and as these are closed under subtraction (that is, for equivalence queries, $h-f$ is also a small read-once oblivious ABP), our hitting set will always contain a distinguishing evaluation for $h$ and $f$, if $h\ne f$.  Thus, we can derandomize the equivalence queries needed for \cite{KlivansShpilka06}.  Further, our results on set-multilinear ABPs, non-commutative ABPs and semi-diagonal depth-4 circuits, all rely on reductions to read-once oblivious ABPs, and these reductions also hold in the learning setting.  We omit further details, and state the following result.

\begin{theorem}[Extension and Derandomization of \cite{KlivansShpilka06}]\label{thm:learning}
	There is a deterministic polynomial-time learning algorithm from membership and equivalence queries that can learn each of the following models: read-once oblivious ABPs, set-multilinear ABPs, non-commutative ABPs and semi-diagonal depth-$4$ circuits. In the first three models the algorithm outputs an hypothesis from the same class as the unknown function, but in the case of semi-diagonal depth-$4$ circuits the outputted hypothesis is a read-once oblivious ABP. Given such query access, the running time of the algorithm is polynomial in the size\footnote{For convenience here, we define the ``size'' of a semi-diagonal depth-4 circuit to be some $\poly(n,k,e,d)$, using the notation of \autoref{thm: semi diagonal}, even though the actual circuit could be much smaller.} of the underlying branching program or circuit.

	Furthermore, such membership and equivalence queries can be implemented in randomized-polynomial time, or in deterministic quasi-polynomial time, for any of the classes mentioned above.
\end{theorem}

\subsection{Related work}\label{sec:related}

\paragraph{Boolean Pseudorandomness:} This work fits into the research program of derandomizing PIT, in particular derandomizing black-box PIT.  However, for many of the models of algebraic circuits studied, there are corresponding boolean circuit models for which derandomization questions can also be asked.  In particular, for a class $\cC$ of boolean circuits, we can seek to construct a \textit{pseudorandom generator} $\cG:\bits^s\to\bits^n$ for $\cC$, such that for any circuit $C\in\cC$ on $n$ inputs, we have the $\epsilon$-closeness of distributions $C(\cG(U_s))\approx_\epsilon C(U_n)$, where $U_k$ denotes the uniform distribution on $\bits^k$. Nisan~\cite{Nisan92} studied pseudorandom generators for space-bounded computation, and for space $S$ computation gave a generator with seed length $s=\O(\lg^2 S)$. Impagliazzo, Nisan and Wigderson~\cite{INW94} later gave a different construction with the same seed length. Randomized space-bounded computation can be modeled with read-once oblivious (boolean) branching programs, and these generators apply to this model of computation as well. 

Our work, at its core, studies read-once oblivious (algebraic) branching programs, and we achieve a quasi-polynomial-sized hitting set, corresponding to a seed length of $\O(\lg^2 S)$ for ABPs of size $S$.  Aside from the similarities in the statements of these results, there are some similarities in the high-level techniques as well. Indeed, an interpretation by Raz and Reingold~\cite{RazReingold99} argues that the \cite{INW94} generator for space-bounded computation can be seen as recursively partitioning the branching program, using an (boolean) extractor to recycle random bits between the partitions.  Similarly, in our work, we use an (algebraic rank) extractor between the partitions.  

However, despite the similarities to the boolean regime, our results improve when the branching programs have bounded width.  Specifically, our hitting sets (\autoref{thm:small width}) achieve a seed length of $\O(\lg^2 S/\lg\lg S)$, and it has been a long-standing open problem (see \cite[Open Problem 8.6]{Vadhan12}) to achieve a seed length (for pseudorandom generators, or even hitting sets) of $o(\lg^2 S)$ for boolean read-once oblivious branching programs of constant width.  Despite much recent work (see \cite{BogdanovDVY13,SimaZ11,GopalanMRTV12,BravermanRRY10,BrodyV10,KouckyNP11,De11,Steinke12}), such seed-lengths are only known for branching programs that are restricted even further, such as regular or permutation branching programs. It is an interesting question as to whether any insights of this paper can achieve a similar seed length in the boolean regime, or in general whether there are any formal connections between algebraic and boolean pseudorandomness for read-once oblivious branching programs.

\paragraph{Derandomizing Noether Normalization:} Mulmuley~\cite{Mulmuley12} (see the full version~\cite{Mulmuley12full}) recently showed that the Noether Normalization Lemma (of commutative algebra and algebraic geometry) can in a sense be made constructive, assuming that PIT can be derandomized.  This lemma shows, roughly, that in any commutative ring $R$, there is a smaller subring $S\subseteq R$ that captures many of the interesting properties of $R$. The usual proof of this lemma is to take the subring $S$ to be generated by a sufficiently large set of ``random'' elements from $R$.  Thus, one can hope to get a constructive version of this lemma by invoking the appropriate derandomization hypothesis from complexity theory, and Mulmuley shows that derandomization of PIT suffices.

A particular focus of Mulmuley's work is when $R$ is the ring of polynomials, whose variables are the $n^2r$ variables in $r$ symbolic $n\times n$ matrices, such that these polynomials are invariant under the simultaneous conjugation of the $r$ matrices by any scalar matrix.  This ring of invariants has an explicit set $T$ of generators, of size $\exp(r,n)$. Noether Normalization shows that there exists a set $T'$ of size $\poly(r,n)$, such that $T'$ generates a subring $S\subseteq R$, and $R$ is \textit{integral}\footnote{A ring $R$ is integral over a subring $S$, if for every $r\in R$, there is some monic polynomial $p(x)\in S[x]$ such that $p(r)=0$.  For example, the algebraic integers are integral over $\Z$, but $\Q$ is not integral over $\Z$.} over $S$.  Mulmuley shows that derandomizing black-box PIT would yield an explicit such set $T'$, and because of the relations with algebraic geometry, conjectures that finding such a set $T'$ is hard.  Mulmuley also derives weaker results, for other rings, just from the assumption that black-box PIT can be derandomized for diagonal circuits.

In this work, we give a quasi-polynomial hitting sets for diagonal circuits (\autoref{thm:diagonal}), making some of Mulmuley's weaker results unconditional.  More interestingly, in follow-up work (\cite{ForbesShpilka13}), we improved Mulmuley's reduction from Noether Normalization to PIT in the case of the above ring $R$ of invariants, and showed that derandomizing PIT for read-once oblivious ABPs is sufficient for finding any explicit set $T'$ of invariants generating the desired subring $S$.  By using the results of this paper (\autoref{thm:main}), one can construct an explicit set $T'$ of size $\poly(n,r)^{\log r}$, despite the conjectured hardness of this problem.

\paragraph{Algebraically Natural Proofs:} Razborov and Rudich~\cite{RazborovRudich97} defined the notion of a \textit{natural proof}, and showed that no natural proof can yield strong lower bounds for many interesting boolean models of computation, assuming widely believed conjectures in cryptography about the existence of pseudorandom functions. Further, they showed that this barrier explains the lack of progress in obtaining such lower bounds, as essentially all of the lower bound techniques known are natural in their sense.

In algebraic complexity, there is also a notion of an algebraically natural proof, and essentially all known lower bounds are natural in this sense (see \cite{Aaronson08} and \cite[\defaultS 3.9]{SY10}).  However, there is no formal evidence to date that algebraically natural proofs cannot prove strong lower bounds, as there are no candidate constructions for (algebraic) pseudorandom functions.  The main known boolean construction, by Goldreich, Goldwasser and Micali~\cite{GGM86}, does not work in the algebraic regime as the construction uses repeated function composition, which results in a polynomial of exponential degree and as such does not fit in the framework of algebraic complexity theory, which only studies polynomials of polynomial degree.

In this work, we give limited informal evidence that some variant of the GGM construction could yield an algebraically natural proofs barrier.  That is, Naor (see \cite{Reingold13}) observed that the GGM construction is superficially similar to Nisan's~\cite{Nisan92} pseudorandom generator, in that they both use recursive trees of function composition.  In this work, we give an algebraic analogue of Nisan's~\cite{Nisan92} boolean pseudorandom generator, and as such use repeated function composition.  As in the naive algebraization of the GGM construction, a naive implementation of our proof strategy would incur a degree blow-up.  While the blow-up would only be quasi-polynomial, it would ultimately result in a hitting set of size $\exp(\lg^3 S)$ for read-once oblivious ABPs of size $S$, instead of the $\exp(\lg^2 S)$ that we achieve.  To obtain this better result, we introduce an interpolation trick that allows us to control the degree blow-up, even though we use repeated function composition.  While this trick alone does not yield the desired algebraic analogue of the GGM construction, it potentially removes the primary obstacle to constructing an algebraic analogue of GGM.

\paragraph{The Partial Derivative Method:} Besides the natural goal of obtaining the most general possible PIT result, the problem of obtaining a black-box version of the algorithm of Raz and Shpilka~\cite{RazShpilka05} is interesting because of the technique used there. Roughly, the PIT algorithm of \cite{RazShpilka05} works for any model of computation that outputs polynomials whose space of partial derivatives is (relatively) low-dimensional. Set-multilinear ABPs, non-commutative ABPs, low-rank tensors, and so called pure algebraic circuits (defined in Nisan and Wigderson~\cite{NisanWigderson96}) are all examples of algebraic models that compute polynomials with that property, and so the algorithm of \cite{RazShpilka05} works for them. In some sense using information on the dimension of partial derivatives of a given polynomial is the most applicable technique when trying to prove lower bounds for algebraic circuits (see e.g., \cite{Nisan91,NisanWigderson96,Raz06,Raz09a,RSY08,RazYehudayoff09}) and so it was an interesting problem to understand whether this powerful technique could be carried over to the black-box setting.  Prior to this work it was not known how to use this low-dimensional property in order to obtain a small hitting set, and this paper achieves this goal. Earlier, in \cite{ForbesShpilka12}, we obtained a black-box algorithm for the model of low-rank tensors, but could not prove that it also works for the more general case of set-multilinear ABPs (we still we do not know whether this is the case or not --- we  discuss the similarities and differences between this and our new algorithm below), so this work closes a gap in our understanding of such low-dimensional models.

\paragraph{Non-commutative ABPs:} While the primary focus of algebraic complexity are polynomials over commutative rings, it has proved challenging to prove lower bounds in sufficiently powerful circuit models because any lower bound would have to exclude any possible ``massive cancellation''.  Thus, in recent years, attention has turned to non-commutative computation, in the hopes that strong lower bounds would be easier to obtain (see \cite{Nisan91,ChienSinclair,ArvindJS09,ArvindS10,HWY,ChienHSS11,HWY2,HrubesY12}).

Similarly, the PIT problem has also been studied in the non-commutative model.  While the work of Raz and Shpilka~\cite{RazShpilka05} establishes a white-box PIT algorithm for non-commutative ABPs, the black-box PIT question for this model is more intricate since one cannot immediately apply the usual Schwartz-Zippel algorithm over non-commutative domains.  However, Bogdanov and Wee~\cite{BogdanovWee05} showed how, leveraging the ideas in the Amitsur-Levitzki theorem~\cite{AmitsurLevitzki}, one can reduce non-commutative black-box PIT questions to commutative black-box PIT questions.  By then appealing to Schwartz-Zippel, they give the first randomized algorithm for non-commutative PIT.  They also discussed the possibility of derandomizing their result and raise a conjecture that if true would lead to a hitting set of size $\approx s^{\lg^2 s}$ for non-commutative ABPs of size $s$.  Our Theorem~\ref{thm:ncABPs} gives a hitting set of size $\approx s^{\lg s}$ and does not require unproven assumptions.

In particular, their conjecture asked about the minimal size of an ABP required to computed a polynomial identity of $n\times n$ matrices. Recall that a polynomial identity of matrices is a non-zero polynomial $f(\vec{x})$ such that no matter which $n\times n$ matrices we substitute for the variables $x_i$, $f$ evaluates to the zero matrix. Our work bypasses this conjecture, as we instead give an improved reduction from non-commutative to commutative computation, such that for ABPs the resulting computation is set-multilinear.  Our construction consists of $(D+1)\times (D+1)$ strictly-upper-triangular matrices for depth $D$ non-commutative ABPs. It is also not hard to see that there is a non-commutative formula of depth $D+1$ which is a polynomial identity  for the space of $(D+1)\times(D+1)$ strictly-upper-triangular matrices. Thus, our result is tight in that it also illustrates that if we go just one dimension up above what is obviously necessary then we can already construct a hitting set.

In \cite{ArvindMS10}, Arvind, Mukhopadhyay and Srinivasan gave a deterministic black-box algorithm for identity testing of {\em sparse} non-commutative polynomials. The algorithm runs in time polynomial in the number of variables, degree and sparsity of the unknown polynomial. This is similar to the running time achieved in the commutative setting for sparse polynomials (see e.g., \cite{Ben-OrTiwari88,KlivansSpielman01}) and in particular it is better than our quasi-polynomial time algorithm. On the other hand our algorithm is more general and works for any non-commutative polynomial that is computed by a small ABP.

We note that in the aforementioned \cite{ArvindMS10}, the authors showed how to deterministically learn sparse non-commutative polynomials in time polynomial in the number of variables, degree and sparsity. In contrast, for such polynomials our deterministic algorithm requires quasi-polynomial time. For general non-commutative ABPs \cite{ArvindMS10} also obtained a deterministic polynomial time learning algorithm, but here they need to have the ability to query the ABP also at internal nodes and not just at the output node. Our deterministic algorithm runs in quasi-polynomial time but it does not need to query the ABP at intermediate computations.  

\paragraph{Previous Work:} In a previous paper \cite{ForbesShpilka12} we gave a hitting set of quasi-polynomial size for the class of depth-$3$ set-multilinear polynomials. That result is mostly subsumed by the generality of Theorem~\ref{thm:main:sm}, but the previous work has a better dependence on some of the parameters.  More importantly, it is interesting to note that the proof in \cite{ForbesShpilka12} is also based on the same intuitive idea of preserving dimension of linear spaces while reducing the number of variables, but in order to prove the results in this paper we had to take a different approach. At a high level the difference can be described as follows. We now summarize the algorithm of \cite{ForbesShpilka12}.
First the case of degree $2$ is solved. In this case the tensor is simply a bilinear map. For larger $D$, the algorithm works by first reducing to the bivariate case using the Kronecker substitution $x_i\leftarrow x^{n^i}$ for $i\leq D/2$ and $x_{i+D/2}\leftarrow y^{n^i}$ for $1\leq i\leq D/2$.  Appealing now to the bivariate case, we can take $y$ to be some multiple of $x$, and then undo the Kronecker substitution (and applying some degree reduction) to recover a tensor on $D/2$ variables.
If we try to implement this approach with ABPs then we immediately run into problems, as the previous work requires that the layers of the ABP are commutative, in that we can re-order the layers.  While this is true for depth-3 set-multilinear computation, it is not true for general ABPs.  To generalize the ideas to general ABPs, this work respects the ordering of the layers in the ABP.  In particular, while the previous work had a top-down recursion, this work follows a bottom-up recursion.  We  merge variables in adjacent levels of the ABP and then reduce the degree of the resulting polynomial using an (algebraic rank) extractor. We do this iteratively until we are left with a univariate polynomial. Perhaps surprisingly, this gives a set that is simpler to describe as compared to the hitting set in \cite{ForbesShpilka12}. On the other hand, if we restrict ourselves to set-multilinear depth-$3$ circuits then the hitting set of \cite{ForbesShpilka12} is polynomially smaller than the set that we construct here.\\

Independently and simultaneously with this work, Agrawal, Saha and Saxena~\cite{AgrawalSS12} obtained results on black-box derandomization of PIT for small-depth set-multilinear formulas, when the partition of the variables into sets is unknown. They obtained a hitting set of size $\exp((2h^2\lg(s))^{h+1})$, for size $s$ formulas of multiplicative-depth $h$. We note that their model is both stronger and weaker compared to the models that we consider here. On the one hand, this model does not assume knowledge of the partition of the variables, whereas our model assumes this knowledge. On the other hand, they only handle small-depth formulas, and indeed, the size of their hitting grows doubly-exponentially in the depth, whereas our results handle arbitrary set-multilinear formulas, according to their definition, if we know the partition, and the size of the hitting set grows quasi-polynomially with the depth\footnote{Their definition of set-multilinear formulas is actually that of a {\em pure} formula (see \cite{NisanWigderson96}). It is not hard to see that pure formulas form a sub-model of the set-multilinear ABPs we examine in this paper. That is, the usual transformation from formulas to ABPs can be done preserving set-multilinearity.}. Using their results for set-multilinear formulas, Agrawal, Saha and Saxena~\cite{AgrawalSS12} also give a quasipolynomial-sized hitting set for semi-diagonal circuits over large characteristic fields.  Our results, in particular our extension of Saxena's~\cite{Saxena08} duality to all fields, can help extend the results of \cite{AgrawalSS12} to small characteristics as well.\\

We also mention that in \cite{JansenQS09,JansenQS10} Jansen, Qiao and Sarma studied black-box PIT in various models related to algebraic branching programs. Essentially, all these models can be cast as a problem of obtaining black-box PIT for read-once oblivious branching programs where each variable appears on a small number of edges. Their result gives a hitting set of size (roughly) $n^{\O(k\lg(k)\lg(n))}$ when $k$ is an upper bound on the number of edges that a variable can label and the ABP is of polynomial size. In comparison, our Theorem~\ref{thm:main} gives a hitting set of size $n^{\O(\lg(n))}$ and works for $k=\poly(n)$ (as long as the size of the ABP is polynomial in $n$). Our techniques are very different from those of \cite{JansenQS09}, which follows the proof technique of \cite{ShpilkaVolkovich09}.  The later paper \cite{JansenQS10} use an algebraic analogue of the Impagliazzo-Nisan-Wigderson~\cite{INW94} generator, as we do, but the details are different.

\subsection{Proof overview}\label{sec:read-once oblivious:overview}

The first step in our proof is to work with a normal form for read-once oblivious ABPs.  Specifically, \autoref{lemma:oabpvsmatrix} shows that any polynomial $f$ computed in this model can be written as\footnote{We use $\zr{n}$ to denote $\{0,\ldots,n-1\}$.} $f(\vec{x})=(\prod_{i\in\zr{D}}M_i(x_i))_{(0,0)}$, where each $M_i$ is a matrix whose entries are univariate polynomials in the variable $x_i$.  Next, we observe that for purposes of a stronger inductive hypothesis, we work with matrix-valued polynomials computed as $f(\vec{x})=\prod_{i\in\zr{D}}M_i(x_i)$.  With this view in mind, we can now consider the recursion scheme.

At a high level, our hitting set can be viewed as a repeated application (in parallel) of a ``derandomized Kronecker product''. That is, the usual Kronecker product allows us to merge two variables (that is, we merge $x$ and $y$ by taking $y\leftarrow x^m$ for large enough $m$) of a polynomial without losing any information. Once a single variable remains, one can resort to full-interpolation of this univariate polynomial.  However, this generic transformation rapidly increases the degree of the polynomial and thus the final interpolation step requires too many evaluations to yield a good hitting set.

For our hitting set, we first observe that in a read-once oblivious ABP each layer has its own variable, so merging two variables in adjacent layers results in the merging of the two layers in the ABP.  We then give a degree-reduction procedure for the results of such merges.  That is, given two adjacent layers of our ABP, $M(x)\in\F[x]^{\zr{r}\times \zr{r}}$ and $N(y)\in\F[y]^{\zr{r}\times \zr{r}}$, we first set $y\leftarrow x^n$ (where $\deg M,N<n$) to obtain $M(x)N(y)\approx M(x)N(x^n)$, where the right-hand term is of degree $\approx n^2$.  We then construct $f,g\in\F[z]$ of degree $<r^2$ such that $M(x)N(x^n)\approx M(f(z))N(g(z))$, where the degree of the right-hand term is now $\approx nr^2$.  Thus we have that $M(x)N(y)\approx M(f(z))N(g(z))$, and we have only increased the degree by $\approx r^2$.  As $r$, the dimension of the matrices (and the width of the ABP), stays constant throughout this merging process, one can repeat this merging process and have the degree of the polynomial scale linearly with the number of mergings. To then fully leverage this idea, we show that pairs of variables can be merged in parallel.

To discuss our degree-reduction process, we first need to state what it means for ``$M(x)N(x^n)\approx M(f(z))N(g(z))$''. In the Kronecker substitution, this equivalence typically means that there is a bijection between monomials.  We will not use that notion, and instead will use a problem-specific notion, related to linear algebra.  Specifically, for each $\alpha,\beta\in\F$, the matrix $M(\alpha)N(\beta)$ is some linear transformation $\F^\zr{r}\to\F^\zr{r}$.  Ranging over all $\alpha$ and $\beta$, we get a space of such transformations, which we extend by linearity to a vector space of transformations.  Crucially, the polynomial $M(x)N(y)$ is zero iff this vector space is zero. More generally, we will show that this linear space of transformations captures essentially all we need to know about the polynomial $M(x)N(y)$. Our degree reduction lemma (using the Kronecker substitution, and then doing the degree reduction) finds two curves $f,g\in\F[z]$ of degree $<r^2$ such that the space of linear transformations induced by $M(x)N(y)$ is the same as the space of linear transformations induced by $M(f(z))N(g(z))$. Namely, the space spanned by  $M(x)N(y)$, where we range over all assignments to $(x,y)$ is the same as the spaces spanned by $M(f(z))N(g(z))$, where we range over all assignments to $z$. It follows then that $M(x)N(y)\approx M(f(z))N(g(z))$, and so we have merged two variables without increasing the degree too much.

To find such curves, we use a (seeded) rank extractor.  That is, the $n^2$ matrices defined by the coefficients of $M(x)N(y)$ live in an $r^2$-dimensional space, so we seek to map these $n^2$ vectors to a smaller space while preserving rank.  Gabizon and Raz~\cite{GabizonRaz08} established such a lemma, and our prior work~\cite{ForbesShpilka12} improved the parameters in their lemma.  Crucially, these rank extractors have the form such that the maps they define correspond to polynomial evaluations.  Thus, these extractors establish (for most values of the seed) an explicit small set of evaluation points where the span of $M(x)N(y)$ is preserved.  

While the strategy above will succeed, it will be deficient as the resulting hitting set will be quasi-polynomially larger than desired, and each point in the hitting set will be quasi-polynomially explicit, instead of being polynomially explicit.  This deficiency arises from repeated function composition, as we now explain by example.  Consider a depth 4, read-once oblivious ABP, written as $M(x,y,z,w)\eqdef A(x)B(y)C(z)D(w)$, where $A,B,C,D$ are $r\times r$ matrices with entries that are univariate polynomials degree $< n$. After the first step we have $A(x)B(y)C(z)D(w) \approx A(f(x'))B(g(x'))C(f(y'))D(g(y'))$ for two new variables $x'$ and $y'$. Viewing this matrix product as a read once oblivious ABP in the two variables $x'$ and $y'$, we can repeat the variable merging procedure, to obtain
\begin{align*}
	A(x)B(y)C(z)D(w)
		&\approx A(f(x'))B(g(x'))C(f(y'))D(g(y'))\\
		&\approx A((f\circ f)(x''))B((g\circ f)(x''))C((f\circ g)(x''))D((g\circ g)(x'')),
\end{align*}
for a new variable $x''$. While this reduction is desirable, as we can now fully interpolate the single final variable $x''$ to obtain a hitting set, the degree in the final $x''$ is $r^2\cdot r^2=r^4$, instead of remaining $r^2$.  That is, to reduce from $D$ variables to a single variable, we will use at least $\lg D$ compositions of the $f,g$ polynomials, yielding a degree of at least $r^{\Omega(\lg D)}$. As there are $\lg D$ levels of recursion, and each requires a seed matching the degree of the ABPs, this will imply that the resulting hitting set is of size $r^{\Omega(\lg^2 D)}$, which is larger than what we achieve in this work.  Further, this hitting set is not even polynomially explicit, as computing it requires evaluating polynomials of quasipolynomial degree, which potentially requires quasipolynomial bit-length when computing over the rationals.

\sloppy Such degree blow-up seems inherent in this naive composition method, so the ``extract and recurse'' paradigm by itself is insufficient for an algebraic analogue of boolean results in space-bounded derandomization.  To overcome this, we avoid treating $M'(x',y')\eqdef A(f(x'))B(g(x'))C(f(y'))D(g(y'))$ as a entirely generic ABP, and recognize that  function composition has occurred, and that we know precisely what those functions are.  That is, the extractor will give a seed $\beta$, and a $\poly(r)$ set of points $P'_\beta\subseteq\F^2$ such that the span of $M'(x',y')$ is equal to the span of $M'|_{P'_\beta}$, for most values of $\beta$.  However, realizing that $x',y'$ define values of $x,y,z,w$ via the known curves $f,g$, we can use the points $P'_\beta$ to construct a new set of points $P_\beta\subseteq\F^4$ such that the span of $M(x,y,z,w)$ is equal to the span of $M|_{P_\beta}$ for most values of $\beta$.  Finally, we can then interpolate curves in a new variable $x''$ to pass through the points $P_\beta$, and these curves will be of degree $|P_\beta|$.  Noting that $|P_\beta|=|P'_\beta|=\poly(r)$, and that this $\poly(r)$ is the number of samples of the extractor, which only depends on the width of the ABPs considered, and this width never changes, we see that there is no degree blow-up throughout the recursion.

\fussy

To achieve better results when the width $r$ of the ABP is small, we keep the paradigm from above, but change the branching factor of the recursion.  That is, in the above recursion we merged pairs of variables at each level, so that if we start with $D$ variables we need $\lg D$ levels of recursion, and the curves involved will have $\poly(r)$ degree.  By changing to a branching factor of $B$, we use $\log_B D$ levels of recursion, and the curves involved will have $\approx \poly(r)^B$ degree. As the number of levels in the recursion equals the number of seeds we use, and the degree of the curves (along with the degree $n$ of the  ABP) governs the number of values to try for each seed, we can see that choosing $B=\log_r D$ yields $\lg D/\lg\lg D\cdot \lg r$ seeds, and each seed will take $\poly(n,D)$ values.  For $r=\O(1)$, this yields a hitting set of size $\poly(n,D)^{O(\lg D/\lg\lg D)}$.

In the boolean regime, changing the branching factor of recursion in the pseudorandom generators of Nisan~\cite{Nisan92} or Impagliazzo-Nisan-Wigderson~\cite{INW94} is not known to achieve such savings. This seems to be because each application of the extractor (viewed as an averaging sampler) has a sample complexity that depends on the total number of variables in the branching program, which is needed so the error does not become too large. It follows that increasing the branching factor simply multiples the seed length by $B/\lg B$, which only becomes worse as $B$ increases.  In contrast, the rank extractor used here has a sample complexity that depends only on the width of the ABP, and the number of total variables only affects the extractor in the number of seeds needed. Thus, the seed-length grows as $(B\lg r+\lg n+\lg D)\lg D/\lg B$, which allows us to balance parameters by taking $B\lg r\approx \lg D$.


\section{Notation}\label{sec:notation}

We write $[n]$ to denote $\{1,\ldots,n\}$, and $\zr{n}$ to denote $\{0,\ldots,n-1\}$. We write $\ind{P}$ to denote the indicator function for the set/event $P$.  We use $\sqcup$ to denote a disjoint union.

Matrices and vectors will most often be indexed starting at zero.  In observance of this indexing scheme, we will write $S^\zr{n}$ to denote $n$-dimensional vectors with entries in $S$, and will write $S^{\zr{n}\times\zr{m}}$ to denote matrices of size $n\times m$ with entries in $S$.  Indexing from 1 (as done in \autoref{sec:diagonal}) will be indicated by the use of $[n]$ as opposed to $\zr{n}$. We denote $\Id$ for the identity matrix, whose dimension will always be clear from the context. For a matrix $M$, we will write $M_{i,\bullet}$ to denote the $i$-th row of $M$, and $M_{\bullet,j}$ to denote the $j$-th column, where $i$ and $j$ are indexed from zero.  Given a bit vector $\vec{b}\in\bits^\zr{n}$, we write $b_i$ for the $i$-th bit of $\vec{b}$, where $i$ is indexed from zero..  When $n=0$, we use $\nulls$ to denote the empty string in $\bits^\zr{0}$, and thus for $n\ge 0$, $|\bits^\zr{n}|=2^n$.

Given a vector of polynomials $\vec{f}\in \F[\vec{x}]^n$ and an exponent vector $\vec{\alpha}\in\N^n$, we write $\vec{f}^{\vec{\alpha}}$  for $f_1^{\alpha_1}\cdots f_n^{\alpha_n}$.

We write $M\in\F[\vec{x}]^{\zr{r}\times \zr{r}}$ to indicate that $M$ is an $r\times r$ matrix, with entries in $\F[\vec{x}]$.  The (total) degree of $M$, written $\deg(M)$, is defined as the maximum degree of its entries.  Given a collection of matrices $\cM\subseteq\F^{\zr{r}\times \zr{r}}$, $\spn\cM$ denotes the span of these matrices, in the vector space $\F^{\zr{r}\times \zr{r}}\equiv \F^\zr{r^2}$. Given $\cM\subseteq\F^{\zr{r}\times \zr{r}}$ and $\cN\subseteq\F^{\zr{r}\times \zr{r}}$, we define $\cM\cdot\cN\eqdef\{MN:M\in\cM, N\in\cN\}$.  Given a polynomial $f\in\F[\vec{x}]$, we write $\coeff{\vec{x}^{\vec{\alpha}}}(f)$ to denote the coefficient of $\vec{x}^{\vec{\alpha}}$ in $f$.  For a matrix $M\in\F[\vec{x}]^{\zr{r}\times \zr{r}}$, we write $\coeff{\vec{x}^{\vec{\alpha}}}(M)$ to denote the $r\times r$ $\F$-matrix, with the $\coeff{\vec{x}^{\vec{\alpha}}}$ operator applied to each entry. When we write ``$f\in\F[\vec{x}][\vec{y}]$'', we will treat $f$ as a polynomial in the variables $\vec{y}$, whose coefficients are polynomials in the variables $\vec{x}$, and correspondingly will write $\coeff{\vec{y}^{\vec{\beta}}}(f)$ to extract the polynomial in $\vec{x}$ that is the coefficient of the monomial $\vec{y}^{\vec{\beta}}$ in $f$. Given a polynomial $f$ (in $\F[\vec{x}]$ or in the non-commutative $\F\{\vec{x}\}$), we write $\homp{k}(f)$ for the homogeneous part of $f$ of degree $k$.

\sloppy If $M_i\in\F^{\zr{r}\times \zr{r}}$ for $i\in\zr{D}$, then we write $\prod_{i\in\zr{D}} M_i$ for the non-commutative multiplication $M_0\cdots M_{D-1}$, with the understanding that the multiplication proceeds left to right, starting at the lower index of the product ($i=0$) and ending at the higher index ($D-1$).  When this non-commutative product is indexed over the range $\bits^\zr{d}$, we order the product by the lexicographic order on $\bits^\zr{d}$, where the least-significant bit is the rightmost bit.

\fussy


\section{Hitting Set for Read-Once Oblivious ABPs}\label{sec:read-once oblivious}

In this section, we construct a hitting set for read-once oblivious ABPs. To make the questions more amenable to study, we give the following normal form for read-once oblivious ABPs.

\begin{lemma}
	\label{lemma:oabpvsmatrix}
	Let $A$ be a width $\le r$, depth $D$, individual degree $<n$ read-once oblivious ABP, computing a polynomial $f\in\F[x_0,\ldots,x_{D-1}]$.  Then for $i\in\zr{D}$, there are matrices $M_i\in\F[x_i]^{\zr{r}\times\zr{r}}$ of degree $<n$ such that \[f(\vec{x})=\left(\prod_{i\in\zr{D}}M_i(x_i)\right)_{(0,0)}\;.\]
	Conversely, any such function can be computed by a width $r$, depth $D$, read-once oblivious ABP.
\end{lemma}
\begin{proof}
	\uline{ABP$\implies$ matrices:} Identify the nodes in the $i$-th layer of $A$ with a subset of $\zr{r}\times\{i\}$, so that the nodes of $A$ are identified with a subset of $\zr{r}\times\zr{D+1}$ and the source is $0\times 0$ and the sink is $0\times D$. Then, define the matrices $M_i$ such that the entry $(M_i)_{\ell,\ell'}$ takes the label of the edge $(\ell\times (i-1),\ell'\times i)$.  Note that such a label is a univariate polynomial in $x_i$ of degree $<n$. It is straight-forward to see that $f$ is computed by this matrix product.

	\uline{matrices$\implies$ ABP:} This is similar to the above.
\end{proof}

Thus, instead of talking about such ABPs, we will be discussing products of matrices $\prod_i M_i(x_i)$.  For induction purposes we will not be restricting ourselves to the $(0,0)$-th entry, but will consider the full matrices.  Such products naturally correspond to multi-source multi-sink ABPs.  

\subsection{Derandomized Kronecker Product for the Span of an ABP}

We will begin the formal statements of this section by giving lemmas on the span of matrices, and how preserving the span of linear transformations can be used for PIT.  We begin by showing that the span of ABPs can be preserved in a recursive fashion.

\begin{lemma}
	\label{lemma:preserving span suffices}
	Let $\cM,\cM'\subseteq \F^{\zr{r}\times \zr{r}}$ and $\cN,\cN'\subseteq\F^{\zr{r}\times \zr{r}}$, such that $\spn\cM=\spn\cM'$ and $\spn\cN=\spn\cN'$.  Then $\spn(\cM\cdot\cN)=\spn(\spn(\cM)\cdot\spn(\cN))$ and $\spn(\cM\cdot\cN)=\spn(\cM'\cdot\cN')$.
\end{lemma}
\begin{proof}
	As $\cM\cdot\cN\subseteq \spn(\cM)\cdot\spn(\cN)$ we see that $\spn(\cM\cdot\cN)\subseteq\spn(\spn(\cM)\cdot\spn(\cN))$ by linearity.  To show the other direction, consider any $A\in \spn(\spn(\cM)\cdot\spn(\cN))$.  Then $A=\sum_i c_i(\sum_j a_{i,j}M_j)(\sum_k b_{i,k}N_k)$, for $M_j\in \cM$ and $N_k\in\cN$.  Thus, by bilinearity of the matrix product, $A=\sum_{i,j,k} c_i a_{i,j}b_{i,k}M_jN_k$, where $M_jN_k\in\cM\cdot\cN$, and thus $A\in\spn(\cM\cdot\cN)$, yielding the first claim.

	The second claim follows from the first, by observing that $\spn(\cM\cdot\cN)=\spn(\spn(\cM)\cdot\spn(\cN))=\spn(\spn(\cM')\cdot\spn(\cN'))=\spn(\cM'\cdot\cN')$.
\end{proof}

We will use this as follows. By \autoref{lemma:oabpvsmatrix} we can restrict our attention to matrix-valued functions of the form $f(\vec{x})=\prod_{i\in\zr{D}} M_i(x_i)$, and we work in this level of generality for ease of induction. As such, $f$ defines a space of matrices $\cA$ by evaluating the polynomial on all possible inputs. Crucially for PIT, we see that $\spn(\cA)$ is zero iff $f$ is zero. By breaking the matrix product in $f$ into variable disjoint left- and right-halves (each on $\approx D/2$ variables), we can see that this space of matrices factorizes into $\cA=\cA_0\cdot\cA_1$, corresponding to the evaluations of the separate halves of the ABP.  Now, because these halves are variable disjoint, we can \textit{in parallel} find smaller spaces of matrices such that $\spn(\cA'_c)=\spn(\cA_c)$ for $c\in\bits$.  From the above lemma, we now see that $\spn(\cA)=\spn(\cA'_0\cdot\cA'_1)$.  The operation ``$\cA'_0\cdot\cA'_1$'' squares the number of matrices under consideration, so to complete the picture we ``derandomize'' this Kronecker product to yield a ``small'' family of matrices $\cA'$ such that $\spn(\cA)=\spn(\cA')$.  It follows then that PIT of the original function $f$ reduces to testing each matrix in $\cA'$ for non-zeroness, which will correspond to fully interpolating a univariate polynomial.

To implement the above program, our main challenge is to thus construct the derandomized Kronecker product.  To do so, our next lemma shows that the span of a layer of a read-once oblivious ABP can be viewed either as the span of the evaluations of that layer, or of the (finite) number of (matrix-valued) coefficients of that layer.  We will use both characterizations of the span of a layer.

\begin{lemma}
	\label{lemma:span of a matrix}
	Let $M\in\F[x]^{\zr{r}\times \zr{r}}$. Then for any set $S\subseteq\F$ with $|S|> \deg(M)$, $\spn\{M(\alpha)\}_{\alpha\in S}=\spn\{\coeff{x^i}(M)\}_{i=0}^{\deg(M)}$.
\end{lemma}
\begin{proof}
	\uline{$\subseteq$:} Note that $M(\alpha)=\sum_{i=0}^{\deg(M)} \coeff{x^i}(M) \cdot \alpha^i$, and thus is in $\spn\{\coeff{x^i}(M)\}_{i=0}^{\deg(M)}$, for any $\alpha$.

	\uline{$\supseteq$:} Let $d=\deg(M)+1$. As the span of a set of vectors is monotonic, it suffices to prove the claim when $|S|=d$. Recall that polynomial interpolation shows that the map $\F^\zr{d}\to\F^\zr{d}$ defined by mapping degree $<d$ polynomials to their evaluations on $S$ is a bijective linear map.  It follows that for any $0\le i<d$, there is a set of constants $\{a_{i,\alpha}\}_{\alpha\in S}$ such that for any polynomial $f\in\F[x]$ of degree $<d$, $\coeff{x^i}(f)=\sum_{\alpha\in S} a_{i,\alpha} f(\alpha)$.  Applying this relation entry-wise we obtain $\coeff{x^i}(M)=\sum_{\alpha\in S} a_{i,\alpha} M(\alpha)$, and thus $\coeff{x^i}(M)\in \spn\{M(\alpha)\}_{\alpha\in S}$.  Varying this over all $i$ yields the result.
\end{proof}

We now cite a dimension reduction lemma that shows how we can map a vector space of rank $\le r$, in a large ambient dimension $n$, to a vector space of the same rank in an ambient dimension $r$ (and thus, is a ``rank extractor'').  A lemma of this sort was first established by Gabizon and Raz~\cite{GabizonRaz08}.  That lemma would also work for us, but the version we cite has slightly better parameters.

\begin{lemma}[\cite{ForbesShpilka12}]\label{lemma:gr}
	Let $1\le s\le r\le n$. Let $M\in\F^{\zr{n}\times \zr{r'}}$ be of rank $s$, for $r'\ge s$.  Let $\K$ be a field extending $\F$, and let $\omega\in\K$ be an element of order $\ge n$.  For $\alpha \in \K$ define $A_\alpha\in\K^{\zr{r}\times \zr{n}}$ by $(A_\alpha)_{i,j}=(\omega^i\alpha)^j$.  Then there are $\le nr-\binom{r+1}{2}<nr$ values $\alpha\in \K$ such that the first $s$ rows of $A_\alpha M$ have rank $<s$.
\end{lemma}


We now apply this dimension reduction to read-once oblivious ABPs, by observing that even though the result $M(x)N(x^n)$ of the Kronecker product induces a linear space of transformations most naturally parameterized by the $\approx n^2$ coefficients of this polynomial (as in \autoref{lemma:span of a matrix}), each transformation is an $r\times r$ matrix, and so this linear space of transformations more naturally lives in an $r^2$-dimensional vector space.  Thus, by applying dimension reduction, we make a first step in getting a derandomized Kronecker product.

\begin{lemma}
	\label{lemma:two var to evals}
	Let $M\in\F[x]^{\zr{r}\times \zr{r}}$, $N\in\F[y]^{\zr{r}\times \zr{r}}$ be of degree $<n$.  Let $\omega$ be an element of order $\ge n^2$. Then, for any $\alpha\in\F$,
	\[\spn\{\coeff{x^iy^j}(M(x)N(y))\}_{i,j\in\zr{n}}\supseteq\spn\{M(\omega^\ell\alpha)N((\omega^\ell\alpha)^n)\}_{\ell\in\zr{r^2}}\]
	and except for $<n^2r^2$ values of $\alpha$,
	\[\spn\{\coeff{x^iy^j}(M(x)N(y))\}_{i,j\in\zr{n}}=\spn\{M(\omega^\ell\alpha)N((\omega^\ell\alpha)^n)\}_{\ell\in\zr{r^2}}\]
\end{lemma}
\begin{proof}
	\sloppy \uline{$\supseteq$:} This follows from the same reasoning used in \autoref{lemma:span of a matrix}, in that $M(\omega^\ell\alpha) N((\omega^{\ell}\alpha)^n)=\sum_{i,j\in\zr{n}} \coeff{x^iy^j} (M(x)N(y))\cdot(\omega^\ell\alpha)^{i+nj}$, and thus the evaluations of $M(x)N(y)$ are in the span of the coefficients of $M(x)N(y)$.

	\uline{$\subseteq$:} Consider the matrix $C\in\F^{\zr{n^2}\times \zr{r^2}}$ defined so that $C_{i+nj,\bullet}\eqdef\coeff{x^iy^j}(M(x)N(y))$, where we flatten $r\times r$ matrices into $r^2$-dimensional vectors.  Observe, that as in the Kronecker substitution, the map $(i,j)\to i+nj$ is a bijection $\zr{n}^2\to\zr{n^2}$, and thus the rows of $C$ are exactly the coefficients of $M(x)N(y)$, and thus $\spn\{\coeff{x^iy^j}(M(x)N(y))\}_{i,j\in\zr{n}}=\rspn(C)$.

	Taking $A_\alpha \in \F^{\zr{r^2}\times \zr{n^2}}$ as defined in \autoref{lemma:gr}, the lemma shows that except for $<n^2r^2$ values of $\alpha$, we have that $\rank(A_\alpha C)=\rank(C)$.  Recall that $(A_\alpha)_{\ell,i+nj}=(\omega^\ell\alpha)^{i+nj}$, which implies that
	\[(A_\alpha)_{\ell,\bullet}C=\sum_{i,j\in\zr{n}}\coeff{x^iy^j}(M(x)N(y))\cdot(\omega^\ell\alpha)^{i+nj}=M(\omega^\ell\alpha)N((\omega^\ell\alpha)^n),\]
	and thus $\spn\{M(\omega^\ell\alpha)N((\omega^\ell\alpha)^n)\}_{\ell\in\zr{r^2}}=\rspn(A_\alpha C)$. Since $\rspn(C)\supseteq\rspn(A_\alpha C)$, it follows that if $\rank(C)=\rank(A_\alpha C)$, then $\rspn(C)=\rspn(A_\alpha C)$. As \autoref{lemma:gr} shows there are $<n^2r^2$ many $\alpha$ for which $\rank(C)>\rank(A_\alpha C)$, and thus there are $<n^2r^2$ many $\alpha$ for which $\rspn(C)\ne\rspn(A_\alpha C)$, yielding the claim.
\end{proof}

The above result shows how to choose evaluation points for the two matrices $M$ and $N$ such that the span is preserved.  However, by itself this does not lead naturally to recursion. Rather, we would like to take $M(x)N(y)$ and produce a pair of curves $f,g\in\F[z]$ satisfying $M(x)N(y)\approx M(f(z))N(g(z))$, such that we can apply recursion and merge the variable $z$ with other variables.  To do this, it is enough that $f$ and $g$ pass through all of the evaluation points for $M$ and $N$ in the above lemma, as then this captures the desired span of linear transformations.  To create these curves, we need the following definition of the Lagrange interpolation polynomials, which will be featured in our hitting set.  Note that these appeared in the prior work of Shpilka and Volkovich~\cite{ShpilkaVolkovich09}, who also used them in a somewhat similar fashion.

\begin{definition}
	Fix $s$, and distinct $\beta_i\in\F$, for $i\in\zr{s}$.  Then the \textbf{Lagrange interpolation polynomials (with respect to $s$ and the $\beta_i$'s)} are the unique polynomials $p_\ell \in\F[t]$ of degree $<s$ such that $p_\ell(\beta_i)=\ind{i=\ell}$ for all $i,\ell\in\zr{s}$.
\end{definition}

We now use these polynomials to construct curves passing through the evaluation points shown in \autoref{lemma:two var to evals}, to yield a reduction of a two-layer read-once oblivious ABP to a one-layer read-once oblivious ABP.  We will phrase the result more generally, so that we can apply induction.


\begin{lemma}
	\label{lemma:two var to one var}
	For $i\in\zr{D}$, let $M_i\in\F[x]^{\zr{r}\times \zr{r}}$, $N_i\in\F[y]^{\zr{r}\times \zr{r}}$ be of degree $<n$, and $f_i\in\F[x]$, $g_i\in\F[y]$ be of degree $\le m$.  Let $\omega\in\F$ be an element of order $\ge (Dnm)^2$. Let $\beta_j\in\F$ with $j\in\zr{r^2}$ be distinct, and let $p_\ell$ be the corresponding Lagrange interpolation polynomials with respect to $r^2$ and the $\beta_j$'s. Then, except for $<(Dnmr)^2$ values of $\alpha$, we have
	\begin{multline*}
		\spn\left\{\prod_{i\in\zr{D}}M_i(f_i(x))\cdot\prod_{i\in\zr{D}}N_i(g_i(y))\right\}_{x,y\in\F}\\
		\subseteq\spn\left\{\prod_{i\in\zr{D}}M_i\left(\sum_{\ell\in\zr{r^2}}f_i(\omega^\ell\alpha)p_\ell(z)\right)\cdot \prod_{i\in\zr{D}}N_i\left(\sum_{\ell\in\zr{r^2}} g_i((\omega^\ell\alpha)^{Dnm})p_\ell(z)\right)\right\}_{z\in\F}
	\end{multline*}.
\end{lemma}
\begin{proof}
	Define $R(x)\eqdef\prod_{i\in\zr{D}}M_i(f_i(x))$, so that $R\in\F[x]^{\zr{r}\times \zr{r}}$ is of degree $<Dnm$, and define $T(y)\eqdef\prod_{i\in\zr{D}}N_i(g_i(y))$, so that $T\in\F[y]^{\zr{r}\times \zr{r}}$ of degree $<Dnm$. As $\omega$ has order $\ge (Dnm)^2$, \autoref{lemma:two var to evals} implies that except for $<(Dnmr)^2$ values of $\alpha$,
	\begin{equation}
		\label{lemma:span generator:eq3}
		\spn\{\coeff{x^iy^j}(R(x)T(y))\}_{i,j\in\zr{Dnm}}
		=\spn\{R(\omega^\ell\alpha)T((\omega^\ell\alpha)^{Dnm})\}_{\ell\in\zr{r^2}} \; .
	\end{equation}
	As $\omega\in\F$ has order $\ge (Dnm)^2$, it follows that $|\F|\ge(Dnm)^2$, so that applying \autoref{lemma:span of a matrix} twice (first to the variable $x$, then to the variable $y$) we have that
	\begin{equation}
		\label{lemma:span generator:eq1}
		\spn\{R(x)T(y)\}_{x,y\in\F}
		=\spn\{\coeff{x^iy^j}(R(x)T(y))\}_{i,j\in\zr{Dnm}}\; .
	\end{equation}
	Now denote \[U(z)\eqdef \prod_{i\in\zr{D}}M_i\left(\sum_{\ell\in\zr{r^2}}f_i(\omega^\ell\alpha)p_\ell(z)\right) \;\; \text{and} \;\; V(z)\eqdef \prod_{i\in\zr{D}}N_i\left(\sum_{\ell\in\zr{r^2}}g_i((\omega^\ell\alpha)^{Dnm})p_\ell(z)\right)\;,\] which are both elements of $\F[z]^{\zr{r}\times \zr{r}}$ of degree $<Dnr^2$. Then by construction of the Lagrange interpolation polynomials, we see that $U(\beta_\ell)=R(\omega^{\ell}\alpha)$ and $V(\beta_\ell)=T((\omega^\ell\alpha)^{Dnm})$, and thus by linearity
	\begin{equation}
		\label{lemma:span generator:eq2}
		\spn\{R(\omega^\ell\alpha)T((\omega^\ell\alpha)^{Dnm})\}_{\ell\in\zr{r^2}}
		\subseteq\spn\{U(z)V(z)\}_{z\in\F}\;.
	\end{equation}
	Putting equations \eqref{lemma:span generator:eq3}, \eqref{lemma:span generator:eq1}, and \eqref{lemma:span generator:eq2} together yields the claim.
\end{proof}

\begin{remark}
	In the above lemma it is tempting to ask for the ``$\subseteq$'' to be an ``$=$'', for in PIT applications we wish to avoid causing a zero polynomial to become a non-zero polynomial (or in this case, avoid causing a zero-dimensional span becoming higher-dimensional).  However, the proof does not give this in general, as inserting the $p_\ell(z)$ polynomial \textit{outside} the functions $f_i$ and $g_i$ allows us to potentially evaluate the matrices $M_i$ and $N_i$ at points possibly outside the range of the polynomials $f_i$ and $g_i$, seemingly allowing the dimension to increase.
	
	Nevertheless, this lemma is enough for constructing a hitting set, as the resulting span is still inside that of $\prod_i M_i(x_i)\cdot \prod_i N_i(y_i)$, which is ultimately the polynomial whose span we are trying to replicate.  In such applications, we will start off (by induction) with $f_i$ and $g_i$ such that \[\spn\left\{\prod_{i\in\zr{D}}M_i(f_i(x))\cdot\prod_{i\in\zr{D}}N_i(g_i(y))\right\}=\spn\left\{\prod_{i\in\zr{D}}M_i(x_i)\cdot\prod_{i\in\zr{D}}N_i(y_i)\right\}\;,\] and so as the ``$\subseteq$'' in \autoref{lemma:two var to one var} nests between the two terms in this equality, we will in fact get ``$=$'' in the lemma in this case, so will not turn zero polynomials into non-zero polynomials.
\end{remark}

\subsection{The Generator for Read-Once Oblivious ABPs}

In this section we construct the hitting set for read-once oblivious ABPs.  The hitting set is more naturally presented as a \textit{generator}, which we now define.  See also the survey~\cite{SY10} for more on generators.

\begin{definition}
	Let $\cC$ be a class of circuits computing polynomials on $n$ variables.  A polynomial map $\cG:\F^\zr{t}\to\F^\zr{n}$ is a \textbf{generator} for $\cC$ if for every polynomial $f$ computed by a circuit in $\cC$, $f\equiv 0 $ iff $f\circ\cG\equiv 0$.
\end{definition}

Given a generator $\cG$, one can show that $\cG(S^t)$ is a hitting set for $\cC$ as long as $|S|$ is polynomially large in the relevant parameters. We present this formally in \autoref{lemma:gen:final}.  Given this relation to hitting sets, we now proceed to the construction of our generator, and then prove the requisite properties.

\begin{construction}
	\label{construction:gen}
	Let $n,r\ge 1$, $d\ge 0$, and $D=2^d$. Let $\F$ be a field of size $|\F|> (Dnr^3)^2$ and let $\omega\in\F$ be of order $\ge (Dnr^2)^2$.  Let $\beta_i\in\F$ with $i\in\zr{r^2}$ be distinct, and let $p_\ell\in\F[t]$ be the Lagrange interpolation polynomials with respect to $r^2$ and the $\beta_i$'s.

	We define a polynomial map $\cG_d:\F^\zr{d+1}\to\F^\zr{D}$ as follows.  We index the inputs to $\cG_d$ by the variables $\alpha_i$ for $i\in\zr{d+1}$, and index the outputs of $\cG_d$ by bit-vectors $\vec{b}\in\bits^\zr{d}$, so that we write $\cG_{d,\vec{b}}$ for the output associated with $\vec{b}$.

	Define,
	\begin{equation}
		\cG_{d,\vec{b}}\eqdef
		\sum_{\ell_0,\ldots,\ell_{d-1}\in\zr{r^2}}
		\prod_{i\in\zr{d}}
		\left((1-b_i)\cdot p_{\ell_{i-1}}(\omega^{\ell_i}\alpha_i)+b_i\cdot p_{\ell_{i-1}}((\omega^{\ell_i}\alpha_i)^{2^inr^2})\right)
		\cdot
		p_{\ell_{d-1}}(\alpha_{d})\;,
	\end{equation}
	where we abuse notation and define $p_{\ell_{-1}}(t)=t$. In particular,
	\begin{equation}
		\cG_{0,\vec{b}}\eqdef p_{\ell_{-1}}(\alpha_0)=\alpha_0.
	\end{equation}
\end{construction}

We now establish properties of this construction. We first prove an easy lemma that gives an upper bound on the degrees of the variables $\alpha_i$ in $\cG$.

\begin{lemma}[Bounding the degree]
	\label{lemma:gen:deg}
	Assume the setup of \autoref{construction:gen}. Let  $\vec{b}\in\bits^\zr{d}$.
	\begin{itemize}
		\item For $i\in\zr{d}$, $\deg_{\alpha_i} (\cG_{d,\vec{b}}(\vec{\alpha}))\le Dnr^4$.
		\item For $i=d$, $\deg_{\alpha_d} (\cG_{d,\vec{b}}(\vec{\alpha}))\le r^2\le Dnr^4$.
	\end{itemize}
\end{lemma}

\begin{proof}
	Recall that $\deg (p_{\ell_i})\le r^2$ (even for $p_{\ell_{-1}}$, as $r\ge 1$).  Also, recall that $2^i\le D$, for $i\in\zr{d+1}$. There are two cases, $0\le i<d$ and $i=d$.  For $0\le i<d$, we see that $\deg_{\alpha_i}(\cG_{d,\vec{b}})\le\deg({p_{\ell_{i-1}}}) \cdot 2^inr^2\le Dnr^4$.  For $i=d$, we see that $\deg_{\alpha_d}(\cG_{d,\vec{b}})\le \deg(p_{\ell_{d-1}})\le r^2 \le Dnr^4$, as $D,n,r\ge 1$.
\end{proof}

The next lemma demonstrates the recursive structure of $\cG$.

\begin{lemma}[The recursive structure]
	\label{lemma:gen:rec}
\sloppy	Assume the setup of \autoref{construction:gen}. For $\vec{b}\in\bits^\zr{d-1}$, $b_{d-1}\in\bits$, and $\vec{\alpha}$ denoting the variables $\alpha_i$ for $i\in\zr{d-1}$,
		\[\cG_{d,\vec{b}b_{d-1}}(\vec{\alpha},\alpha_{d-1},\alpha_d)
			=\begin{cases}
				\sum_{\ell_{d-1}\in\zr{r^2}}\cG_{d-1,\vec{b}}(\vec{\alpha},\omega^{\ell_{d-1}}\alpha_{d-1})p_{\ell_{d-1}}(\alpha_d)	&	\text{if } b_{d-1}=0\\
				\sum_{\ell_{d-1}\in\zr{r^2}}\cG_{d-1,\vec{b}}(\vec{\alpha},(\omega^{\ell_{d-1}}\alpha_{d-1})^{Dnr^2/2})p_{\ell_{d-1}}(\alpha_d)	&	\text{else}\\
			\end{cases}\;\;.
		\]
\end{lemma}
\begin{proof}
	This claim follows from the definitions, with some rearrangement of terms.
	\begin{align*}
		&\cG_{d,\vec{b}b_{d-1}}(\vec{\alpha},\alpha_{d-1},\alpha_d)\\
		&=\sum_{\ell_0,\ldots,\ell_{d-1}\in\zr{r^2}} \prod_{i\in\zr{d}}\left((1-b_i)\cdot p_{\ell_{i-1}}(\omega^{\ell_i}\alpha_i)+b_i\cdot p_{\ell_{i-1}}((\omega^{\ell_i}\alpha_i)^{2^inr^2})\right) \cdot p_{\ell_{d-1}}(\alpha_{d})\\
		&=\sum_{\ell_{d-1}\in\zr{r^2}}\sum_{\substack{\ell_i\in\zr{r^2}\\i\in\zr{d-1}}} \prod_{i\in\zr{d-1}}\left((1-b_i)\cdot p_{\ell_{i-1}}(\omega^{\ell_i}\alpha_i)+b_i\cdot p_{\ell_{i-1}}((\omega^{\ell_i}\alpha_i)^{2^inr^2})\right)\\
		&\hspace{.5in}\cdot\left((1-b_{d-1})\cdot p_{\ell_{d-2}}(\omega^{\ell_{d-1}}\alpha_{d-1})+b_{d-1}\cdot p_{\ell_{d-2}}((\omega^{\ell_{d-1}}\alpha_{d-1})^{2^{d-1}nr^2})\right)\cdot p_{\ell_{d-1}}(\alpha_{d})\\
		&=\sum_{\ell_{d-1}\in\zr{r^2}}((1-b_{d-1})\cdot\cG_{d-1,\vec{b}}(\vec{\alpha},\omega^{\ell_{d-1}}\alpha_{d-1})+b_{d-1}\cdot\cG_{d-1,\vec{b}}(\vec{\alpha},(\omega^{\ell_{d-1}}\alpha_{d-1})^{2^{d-1}nr^2})\cdot p_{\ell_{d-1}}(\alpha_{d})\\
		&=\begin{cases}
			\sum_{\ell_{d-1}\in\zr{r^2}}\cG_{d-1,\vec{b}}(\vec{\alpha},\omega^{\ell_{d-1}}\alpha_{d-1})p_{\ell_{d-1}}(\alpha_d)	&	\text{if } b_{d-1}=0\\
			\sum_{\ell_{d-1}\in\zr{r^2}}\cG_{d-1,\vec{b}}(\vec{\alpha},(\omega^{\ell_{d-1}}\alpha_{d-1})^{Dnr^2/2})p_{\ell_{d-1}}(\alpha_d)	&	\text{else}
		\end{cases}
		\qedhere
	\end{align*}
\end{proof}

We now show that the generator can be efficiently computed.

\begin{lemma}[Efficiency]
	\label{lemma:gen:compute}
	Assume the setup of \autoref{construction:gen}. For $\vec{b}\in\bits^\zr{d}$, $\cG_{d,\vec{b}}:\F^\zr{d+1}\to\F^\zr{D}$ is computable by a depth $d+1$, width $r^2$, read-once oblivious ABP\footnote{That is, we are constructing a generator for read-once oblivious ABPs, and each coordinate is computed by a read-once oblivious ABP on many fewer variables.  It is unclear at the moment whether this is amenable to recursion, since read-once oblivious ABPs are not closed under composition.} in the variables $\vec{\alpha}$, of degree $\le Dnr^4$.
\end{lemma}
\begin{proof}
	For $i\in\zr{d+1}$, and $c\in\bits$ define $A_{i,c}\in \F[\alpha_i]^{\zr{r^2}\times \zr{r^2}}$, by
	\begin{align*}
		(A_{i,0})_{\ell_{i-1},\ell_{i}}(\alpha_i)&\eqdef p_{\ell_{i-1}}(\omega^{\ell_i}\alpha_i), &
		(A_{i,1})_{\ell_{i-1},\ell_{i}}(\alpha_i)&\eqdef p_{\ell_{i-1}}((\omega^{\ell_i}\alpha_i)^{2^inr^2})
	\end{align*}
	for $\ell_{i-1},\ell_i\in\zr{r^2}$, where $\ell_{-1}$ and $\ell_{d}$ are ranged over as well, and $p_{\ell_{-1}}$ is still the identity polynomial.  Then we have that
	\begin{align*}
		\cG_{d,\vec{b}}(\vec{\alpha})
		&=\sum_{\ell_0,\ldots,\ell_{d-1}\in\zr{r^2}} \prod_{i\in\zr{d}} \left((1-b_i)\cdot p_{\ell_{i-1}}(\omega^{\ell_i}\alpha_i)+b_i\cdot p_{\ell_{i-1}}((\omega^{\ell_i}\alpha_i)^{2^inr^2})\right) \cdot p_{\ell_{d-1}}(\alpha_{d})\\
		&=\sum_{\substack{\ell_{-1}=\ell_d=0\\\ell_0,\ldots,\ell_{d-1}\in\zr{r^2}}} \prod_{i\in\zr{d}} \left((1-b_i)\cdot A_{i,0}(\alpha_i)_{\ell_{i-1},\ell_i}+b_i\cdot A_{i,1}(\alpha_i)_{\ell_{i-1},\ell_i})\right) \cdot A_{d,0}(\alpha_d)_{\ell_{d-1},\ell_d}\\
		&=\left(\left[\prod_{i\in\zr{d}}\left((1-b_i)\cdot A_{i,0}(\alpha_i)+b_i\cdot A_{i,1}(\alpha_i)\right)\right]\cdot A_{d,0}(\alpha_d)\right)_{0,0}
	\end{align*}
	by the properties of matrix multiplication.  By \autoref{lemma:oabpvsmatrix} we see that this matrix product can be computed by a width $r^2$, depth $d+1$, read-once oblivious ABP, and the degree bound follows from \autoref{lemma:gen:deg}.
\end{proof}

We now turn to proving that our construction is indeed a generator for read-once oblivious ABPs.  In particular, we first show the generator preserves span.

\begin{lemma}[Span preserving]
	\label{lemma:gen:span}
	Assume the setup of \autoref{construction:gen}.  Let $M_{\vec{b}}\in\F[x_{\vec{b}}]^{\zr{r}\times \zr{r}}$ for $\vec{b}\in\bits^\zr{d}$ be of degree $<n$.  Then, denoting $\vec{\alpha}$ for the vector of variables $\alpha_0,\ldots,\alpha_d$,
	\begin{equation}
		\label{theorem:gen:eq}
		\spn\left\{\prod_{\vec{b}\in\bits^\zr{d}} M_{\vec{b}}\left(x_{\vec{b}}\right)\right\}_{\vec{x}\in\F^\zr{D}}
		=\spn\left\{\prod_{\vec{b}\in\bits^\zr{d}} M_{\vec{b}}\left(\cG_{d,\vec{b}}(\vec{\alpha})\right)\right\}_{\vec{\alpha}\in\F^\zr{d+1}}\;.
	\end{equation}
\end{lemma}
\begin{proof}
	We prove the claim by induction on $d$.

	\paragraph{\uline{$d=0$:}} Observe that for $d=0$, $\bits^\zr{d}=\{\nulls\}$, where $\nulls$ is the empty string.  Further, we have that $\cG_{0,\nulls}:\F\to\F$ is simply the identity map, by definition of $p_{\ell_{-1}}$.  Thus, \autoref{theorem:gen:eq} is stating that, \[\spn\{M_\nulls(x_\nulls)\}_{x_\nulls\in\F}=\spn\{M_\nulls(\cG_{0,\nulls}(\alpha_0))\}_{\alpha_0\in\F}=\spn\{M_\nulls(\alpha_0)\}_{\alpha_0\in\F}\] and is thus trivially true.

	\paragraph{\uline{$d>0$:}} We will split the span into the multiplication of two different spans, and then appeal to \autoref{lemma:two var to one var} to merge these two spans.
	
	Specifically, denote
	\[\cM\eqdef \left\{\prod_{\vec{b}\in\bits^\zr{d}} M_{\vec{b}}\left(x_{\vec{b}}\right)\right\}_{\vec{x}\in\F^\zr{D}}\]
	and for $c\in\bits$ denote
	\[\cM_c\eqdef \left\{\prod_{\vec{b}\in\bits^\zr{d-1}} M_{\vec{b}c}\left(x_{\vec{b}c}\right)\right\}_{\vec{x}\in\F^\zr{D/2}}\;.\]
	Similarly, for $c\in\bits$, and $\vec{\alpha}$ being the vector of variables $\alpha_i$ for $i\in\zr{d-1}$, denote
	\[\cM'\eqdef \left\{\prod_{\vec{b}\in\bits^\zr{d}} M_{\vec{b}}\left(\cG_{d,\vec{b}}\left(\vec{\alpha},\alpha_{d-1},\alpha_d\right)\right)\right\}_{\vec{\alpha}\alpha_{d-1}\alpha_d\in\F^\zr{d+1}}\]
	and
	\[\cM'_c\eqdef \left\{\prod_{\vec{b}\in\bits^\zr{d-1}} M_{\vec{b}c}\left(\cG_{d-1,\vec{b}}\left(\vec{\alpha},\alpha_{d-1}\right)\right)\right\}_{\vec{\alpha}\alpha_{d-1}\in\F^\zr{d}}\;.\]
	It follows from definition that $\cM=\cM_0\cdot\cM_1$, and by induction, we have that $\spn\cM_c=\spn\cM'_c$.  Thus, \autoref{lemma:preserving span suffices} implies that $\spn\cM=\spn(\spn(\cM'_0)\cdot\spn(\cM'_1))=\spn(\cM'_0\cdot\cM'_1)$.  Observe that $\spn\cM'\subseteq \spn\cM=\spn(\cM'_0\cdot\cM'_1)$, so the claim will follow from showing that $\spn(\cM'_0\cdot\cM'_1)\subseteq\spn\cM'$.

	We now show that $\spn(\cM'_0\cdot\cM'_1)\subseteq\spn\cM'$. We will prove this claim for each fixed value of the variables $\vec{\alpha}$.  Thus, the matrices in each of $\cM'_c$ are parameterized only by the variable $\alpha_{d-1}$.  Further, each matrix in $\cM'_c$ is the product of $D/2$ matrices each of the form $M(f(\alpha_{d-1}))$, for $M\in\F[\alpha_{d-1}]^{\zr{r}\times \zr{r}}$ of degree $<n$ and $f\in\F[\alpha_{d-1}]$ of degree $\le r^2$ (where the degree bound is by \autoref{lemma:gen:deg}).  As the order of $\omega$ is $\ge (Dnm)^2$, and $|\F|\ge (Dnmr)^2$ (taking $m=r^2$) we can apply \autoref{lemma:two var to one var} to see that there is some\footnote{\autoref{lemma:two var to one var} shows that there are $<(Dnmr)^2=D^2n^2r^6$ possible bad values of $\hat{\alpha}_{d-1}$, but in \autoref{lemma:gen:final} we show that there are ``really'' only $Dn^2r^4$ many bad values.  It is unclear at present how to directly improve \autoref{lemma:two var to one var} to reflect this better bound.} value $\hat{\alpha}_{d-1}\in\F$ such that
	\begin{align*}
		&\spn\left\{\prod_{\vec{b}\in\bits^\zr{d-1}} M_{\vec{b}0}\left(\cG_{d-1,\vec{b}}\left(\vec{\alpha},\alpha_{d-1}\right)\right)
			\cdot \prod_{\vec{b}\in\bits^\zr{d-1}} M_{\vec{b}1}\left(\cG_{d-1,\vec{b}}\left(\vec{\alpha},\alpha_{d-1}\right)\right)
			\right\}_{\alpha_{d-1}\in\F}\\
		&\hspace{0.5in}
			\subseteq \spn\left\{\prod_{\vec{b}\in\bits^\zr{d-1}} M_{\vec{b}0}\left(\sum_{\ell_{d-1}\in\zr{r^2}}\cG_{d-1,\vec{b}}\left(\vec{\alpha},\omega^{\ell_{d-1}}\hat{\alpha}_{d-1}\right)p_{\ell_{d-1}}(\alpha_d)\right)\right.\\
		&\hspace{1.5in}
			\left.\cdot \prod_{\vec{b}\in\bits^\zr{d-1}} M_{\vec{b}1}\left(\sum_{\ell_{d-1}\in\zr{r^2}}\cG_{d-1,\vec{b}}\left(\vec{\alpha},(\omega^{\ell_{d-1}}\hat{\alpha}_{d-1})^{Dnr^2/2}\right)p_{\ell_{d-1}}(\alpha_d)\right)\right\}_{\alpha_d\in\F}
		\intertext{and rewriting this via \autoref{lemma:gen:rec},}
		&\hspace{1in}
			=\spn\left\{\prod_{\vec{b}\in\bits^\zr{d-1},c\in\bits} M_{\vec{b}c}\left(\cG_{d,\vec{b}c}(\vec{\alpha},\hat{\alpha}_{d-1},\alpha_d)\right)\right\}_{\alpha_d\in\F}\\
		&\hspace{1in}
			\subseteq\spn\left\{\prod_{\vec{b}\in\bits^\zr{d-1},c\in\bits} M_{\vec{b}c}\left(\cG_{d,\vec{b}c}(\vec{\alpha},\alpha_{d-1},\alpha_d)\right)\right\}_{\alpha_{d-1},\alpha_d\in\F}\;\;.
	\end{align*}
	Thus taking this for each $\vec{\alpha}$ we get $\spn(\cM'_0\cdot\cM'_1)\subseteq\spn\cM'$, implying $\spn\cM'=\spn\cM$ as desired.
\end{proof}

The next lemma concludes that the span-preservation property shows that we indeed have a generator, and can thus construct a hitting set.

\begin{lemma}
	\label{lemma:gen:final}
	Assume the setup of \autoref{construction:gen}.  Let $M_{\vec{b}}\in\F[x_{\vec{b}}]^{\zr{r}\times \zr{r}}$ for $\vec{b}\in\bits^\zr{d}$ be of degree $<n$, and let $f(\vec{x})=\left(\prod_{b\in\bits^\zr{d}} M_{\vec{b}}(x_{\vec{b}})\right)_{0,0}$.  Let $S\subseteq\F$, with $|S|\ge Dn^2r^4$.  Then $f\equiv 0$ iff $f\circ \cG_d\equiv 0$ iff $f|_{\cG_d(S^{d+1})}\equiv 0$.
\end{lemma}
\begin{proof}
	As $f$ is the $(0,0)$-entry of the matrix product $\prod_{b\in\bits^\zr{d}} M_{\vec{b}}(x_{\vec{b}})$, and this projection operator is linear, we see that \autoref{lemma:gen:span} implies that $\spn\{f(\vec{x})\}_{\vec{x}\in\F^\zr{D}}=\spn\{f(\cG_d(\vec{\alpha}))\}_{\vec{\alpha}\in\F^\zr{d+1}}$, and as these spans are simply constant multiples of the output of $f$, we see that $f\equiv 0$ iff $f\circ \cG_d\equiv 0$.  Next, we observe that for $i\in\zr{d+1}$, $\deg_{\alpha_i} (f\circ \cG_d)<Dn^2r^4$ using that $f$ has individual degrees $<n$, and invoking the degree bound in \autoref{lemma:gen:deg}.  Thus, by basic (multivariate) polynomial interpolation, $f\circ \cG_d\equiv 0$ iff $(f\circ \cG_d)|_{S^{d+1}}\equiv 0$ iff $f|_{\cG_d(S^{d+1})}\equiv 0$.
\end{proof}

We can now prove our main theorem.

\begin{theorem}[PIT for read-once oblivious ABPs]\label{thm:main}
	Let $\mathcal{C}$ be the set of polynomials $f:\F^\zr{D}\to\F$ computable by a width $r$, depth $D$, individual degree $<n$ read-once oblivious ABP.  If $|\F|\ge (2Dnr^3)^2$, then $\mathcal{C}$ has a $\poly(D,n,r)$-explicit hitting set, of size $\le (2Dn^2r^4)^{\lceil \lg D\rceil +1}$.
\end{theorem}
\begin{proof}
	By \autoref{lemma:oabpvsmatrix}, any such read-once oblivious ABP computes a polynomial $f$ expressible as $f(\vec{x})=(\prod_{i\in\zr{D}}M_i(x_i))_{(0,0)}$, where $M_i(x_i)\in\F[x_i]^{\zr{r}\times\zr{r}}$ is of degree $<n$. Further, we can round $D$ up to $2^{\lceil \lg D\rceil}$, and by setting $M_i(x_i)=\Id$ for $i\ge D$ (ie.\ padding the product with identity matrices in new variables), we get that $f(\vec{x})=(\prod_{i\in\zr{2^{\lceil \lg D\rceil}}} M_i(x_i))_{0,0}$.
	
	We now show that there is some element $\omega\in\F$ of order $\ge(2^{\lceil \lg D\rceil}nr^2)^2$, that can be found in $\poly(D,n,r)$-time.  As there at most $d$ elements such that $x^d=1$ in $\F$, it follows that there are at most $<(2Dnr^2)^4$ many non-zero elements of order $<(2Dnr^2)^2$, so any enumeration of non-zero elements in $\F$ will find such an element of large order, if $(2Dnr^2)^4<|\F|$.  If $(2Dnr^2)^4\ge |\F|$, then recalling that the multiplicative group of a finite field $\F$ is cyclic, we see there is an element $g\in\F$ of order $|\F|-1\ge (2Dnr^3)^2-1\ge (2^{\lceil \lg D\rceil}nr^3)^2$ (using that $2D>2^{\lceil \lg D\rceil}$), and thus an enumeration can also find such an $\omega$.  Similarly, we can find $r^2$ distinct $\beta_i$, and can find some set $S\subseteq\F$ of size $2^{\lceil \lg D\rceil}n^2r^4$.  With $\omega$, the $\beta_i$'s and $S$, we see then, by \autoref{lemma:gen:final}, that $\cG_d(S^{\lceil \lg D\rceil+1})$ is a hitting set for $f$ and has the desired size.  Further, for any fixed $\vec{s}\in S^{\lceil \lg D\rceil+1}$, $\cG_d(\vec{s})$ can be computed in $\poly(D,n,r)$ time as each coordinate is computed by a small ABP, as shown in \autoref{lemma:gen:compute}.
\end{proof}

We note that Theorem~\ref{thm:main:sm} is an immediate corollary.

\begin{theorem}[PIT for set-multilinear ABPs]\label{thm:main:sm}
	Let $\vec{X}=\sqcup_{i\in\zr{D}}\vec{X}_i$ where $\vec{X}_i = \{x_{i,j}\}_{j\in\zr{n}}$ be a known partition. Let $\mathcal{C}$ be the set of set-multilinear polynomials $f(\sqcup_{i\in\zr{D}}\vec{X}_i):\F^\zr{nD}\to\F$ computable by a width $r$, depth $D$, set-multilinear ABP. If $|\F|\ge (2Dnr^3)^2$, then $\mathcal{C}$ has a $\poly(D,n,r)$-explicit hitting set, of size $\le (2Dn^2r^4)^{\lceil \lg D\rceil +1}$.
\end{theorem}
\begin{proof}
	By the degree bounds, the Kronecker map $x_{i,j}\leftarrow x_i^j$ induces a bijection among the monomials, and takes the original set-multilinear ABP to a read-once oblivious ABP, where we replace the set of variables $\vec{X}_i$ with the single (new) variable $x_i$.  Structurally, this is still an width $r$, depth $D$ ABP, and as we index from zero, the individual degrees are $<n$. Appealing to \autoref{thm:main} gives the result.
\end{proof}


\subsection{Small width read-once oblivious ABPs}

The proof of \autoref{thm:main} used a recursion scheme that merges pairs of variables in each level of the recursion.  By merging variables in larger groups, we can achieve better results when the width $r$ of the ABP is small.

\begin{theorem}[PIT for small width read-once oblivious ABPs]\label{thm:small width}
	\sloppy	Let $\cC$ be the set of $D$-variate polynomials computable by width $r\le \O(1)$, depth $D$, individual degree $<n$ read-once oblivious ABPs.  If $|\F|\ge \poly(D,n)$, then $\mathcal{C}$ has a $\poly(D,n)$-explicit hitting set, of size $\le \poly(D,n)^{\O(\lg D/\lg\lg D)}$.
\end{theorem}
\begin{proof}[Proof sketch]
	The details are similar to that of \autoref{thm:main}, except that instead of merging pairs of variables at once, we merge $B=\log_r D$ variables at once.  The merging of \autoref{thm:main} first merges variables via the Kronecker product, and then reduces to a single variable via the rank extractor.  However, this would create degrees in the $\alpha$'s of $\ge(Dn)^B$, which would not give any savings in the $r\le \O(1)$ case.

	Instead, to merge $B$ variables together, we first apply the rank extractor to each variable to find $\poly(r)$ (seeded) points that preserve the span.  We can then pass curves in a new variable through the Cartesian product of these points, and this will yield a new ABP in a single variable with the same span.  As constructed, this requires curves of degree $\poly(r)^B$ in the generator, and when composed with the ABP itself will induce degrees of only $\poly(D,n,r^B)$ in the seeds $\alpha$.  By using this new merging scheme in our recursion, it follows then that when $r\le\O(1)$ we can take $B=\log_r D$ and the $\alpha_i$ will all have $\poly(n,D)$-degree, and we have $\log_BD$ such $\alpha_i$, resulting in a hitting set of size $\poly(D,n)^{\O(\lg D/\lg\lg D)}$.
\end{proof}


\section{Non-commutative ABPs}
\label{sec: ncABP}


In this section, we show how to do black-box identity testing of non-commutative ABPs in quasi-polynomial time.  A non-commutative ABP is structurally the same as defined in \autoref{def: ABP}, but we now do not assume the variables commute.  Formally, it is an ABP over the ring $\F\{\vec{x}\}$, of polynomials in non-commuting variables. Nisan~\cite{Nisan91} first explored this model, noting that any commutative polynomial $f$ can be computed non-commutatively (for example, by the complete monomial expansion) and thus non-commutative models of computation form a restricted class of computation we can explore.  In the same work, he proved the first exponential lower bounds in the non-commutative ABP model for the (commutative) determinant and permanent polynomials.  Later, Raz and Shpilka~\cite{RazShpilka05} gave a polynomial time PIT algorithm in the white-box model for this class.

In both of the above works, non-commutativity can be seen as ``syntactic'' in the sense that one can treat the variables as formally non-commutating free variables, and one doesn't seek to substitute values for these variables.  However, black-box PIT of non-commutative polynomials in $\F\{\vec{x}\}$ by definition requires such a substitution.  Such a substitution will occur in some ring $\cR$ that is an $\F$-algebra (that is, there is a (commutative) ring homomorphism\footnote{We need this homomorphism to make sense of how elements of $\F$ act on elements of $\cR$.  That is, when evaluating a non-commutative polynomial $f$ in $\F\{\vec{x}\}$ on elements of $\cR$, we replace the coefficients in $f$ with their images under this homomorphism.} $\F\to\cR$).  Clearly, $\cR$ must be non-commutative in order to witness the non-identity $xy-yx$ in $\F\{\vec{x}\}$.  Trivially, we could choose ``$\cR=\F\{\vec{x}\}$'' and get a black-box PIT algorithm that only requires a single query, but this simply pushes the complexity of the problem into the operations in $\cR$.  That is, one would want the black-box PIT queries to be efficiently implementable given white-box access to the circuit.  As such, it seems natural to ask that $\cR$ is a finite dimensional $\F$-vector space, and that the number of dimensions is polynomial in the relevant parameters.  Further, another natural restriction is to take $\cR=\F^{\zr{m}\times\zr{m}}$, the algebra of matrices\footnote{The relevant homomorphism $\F\to\F^{\zr{m}\times\zr{m}}$ maps $a\mapsto a\Id$, where $\Id$ is the identity matrix.}, and that is what we do here.


However, once we have moved from evaluations in $\F\{\vec{x}\}$ to evaluations in $\cR$, there is the concern that we have lost information, in that $f\in\F\{\vec{x}\}$ could vanish on all possible evaluations over $\cR$.  Note that this is also a problem in the commutative case, as in the standard example of the polynomial $x^2-x\in\F_2[x]$ which vanishes over the field $\F_2$.  In the case of matrices, which is the ring we shall work over, there are such identities given by the Amitsur-Levitzki theorem~\cite{AmitsurLevitzki}: the polynomial $\sum_{\sigma\in S_{2m}} \sgn(\sigma) \prod_{i\in\zr{2m}} A_{\sigma(i)}$, where $S_{2m}$ is the symmetric group on $2m$ elements, vanishes on every choice of $2m$ $m\times m$ matrices $A_i$, over any field.

However, recall that in the commutative case, (multivariate) polynomial interpolation states that for a polynomial $f\in\F[\vec{x}]$ of (total) degree $<d$, $f$ cannot vanish on all evaluations over $\F$ as long as $|\F|\ge d$.  Extending this, the Schwartz-Zippel lemma shows that if $|\F|\gg d$ then $f$ has a very low probability of vanishing on a random evaluation over $\F$.  This result, applied via a union bound in a probabilistic argument, shows that efficient black-box PIT is (existentially) possible for small ABPs.  In almost direct analogy, the Amitsur-Levitzki theorem shows that polynomials of (total) degree $<2m$ in $\F\{\vec{x}\}$ cannot be identities over $\F^{\zr{m}\times\zr{m}}$.  Bogdanov and Wee~\cite{BogdanovWee05} observed that this result, in combination with the Schwartz-Zippel lemma, show that if $|\F|,m\gg d$ then a non-commutative polynomial $f$ has a very low probability of vanishing on a random evaluation over $\F^{\zr{m}\times\zr{m}}$.  And thus, as in the commutative case, this establishes existence of small hitting sets for non-commutative polynomials computed by small ABPs.  Given this existential argument, we now proceed to construct quasi-polynomial sized, explicit hitting sets, for non-commutative ABPs.

We proceed, 
by giving a family of matrices with entries consisting of commutative variables, such that any non-zero non-commutative polynomial $f$ over these matrices induces a non-zero family of commutative polynomials over the commutative variables. Further, we show that if $f$ is computable by a small non-commutative ABP, then the resulting commutative polynomials are computable by small set-multilinear ABPs.  We can then appeal to the hitting set for this class of polynomials, shown in \autoref{thm:main:sm}.  We first illustrate the construction of our matrices.

\begin{construction}
	\label{construction:nc-to-sm}
	Define $\varphi:\F\{x_i\}_{i\in\zr{n}}\to\F[x_{i,j}]_{i\in\zr{n},j\in\N}$ to be the unique $\F$-linear map defined by $\varphi(1)=1$, and
	\[\varphi(x_{i_1}x_{i_2}\cdots x_{i_d})\eqdef x_{i_1,1}x_{i_2,2}\cdots x_{i_d,d}\;.\]
	For $i\in\zr{n}$, $D\ge0$, define the upper-triangular matrix $X_{i,D}\in(\F[x_{i,j}]_{j\in\N})^{\zr{D+1}\times\zr{D+1}}$ by
	\[(X_{i,D})_{k,\ell}
		\eqdef\begin{cases}
			x_{i,\ell}&	\text{if } k+1=\ell\\
			0	&	\text{else}
		\end{cases}
		\;.
	\]
	Define the vector $\vec{X}_D\eqdef(X_{i,D})_{i\in\zr{n}}$.
\end{construction}

Thus, pictorially, we have that the matrices $X_{i,D}$ look as follows.
\[X_{i,D}=
	\begin{bmatrix}
		0	&x_{i,1}	&0		&0		&\cdots		&0	\\
		0	&0		&x_{i,2}	&0		&\cdots		&0	\\
		0	&0		&\ddots		&\ddots		&\vdots		&\vdots\\
		0	&0		&\cdots		&0		&x_{i,D-1}	&0\\
		0	&0		&\cdots		&0		&0		&x_{i,D}\\
		0	&0		&\cdots		&0		&0		&0\\
	\end{bmatrix}
	\;.
\]
These matrices are quite similar to the matrices used in \cite{BogdanovWee05} (see their Lemma~3.2), but they achieve a smaller dimension (in fact, matching the best possible, as seen by the Amitsur-Levitzki theorem).  However, their construction does not seem to give the reduction to set-multilinear computation that we need.  We now establish properties of our construction.  In particular, we relate the evaluations of $f$ on the matrices $X_i$, to the commutative polynomial $\varphi(f)$.

\begin{lemma}
	\label{lemma:staircase}
	Assume the setup of \autoref{construction:nc-to-sm}.
	Let $f\in\F\{x_i\}_{i\in\zr{n}}$ be a non-commutative polynomial.  For $D\ge \ell\ge 0$,
	\begin{equation*}
		(f(\vec{X}_D))_{0,\ell}=\varphi(\homp{\ell}(f))\;.
	\end{equation*}
\end{lemma}
\begin{proof}
	By $\F$-linearity (of $\varphi$, $\homp{\ell}(\bullet)$, and $(\bullet)_{0,\ell}$), it is enough to prove the claim for monomials $f(\vec{x})=x_{i_1}x_{i_2}\cdots x_{i_d}$ for any $d\ge 0$.  For $d=0$, the only monomial is the identity matrix $\Id$, for which the claim is clear.
	
	Now consider $d>0$. Observe that
	\begin{align*}
		(X_{i_1}\cdots X_{i_d})_{0,\ell}
		&=\sum_{\substack{\ell_0=0,\ell_d=\ell\\\ell_1,\ldots,\ell_{d-1}\in\zr{D+1}}} (X_{i_1})_{\ell_0,\ell_1}(X_{i_2})_{\ell_1,\ell_2}\cdots(X_{i_d})_{\ell_{d-1},\ell_d}
		\intertext{and by construction, we see that if $\ell_{j-1}+1\ne\ell_{j}$ then $(X_{i})_{\ell_{j-1},\ell_j}=0$, and so,}
		&=\sum_{\substack{\ell_0=0,\ell_d=\ell\\\ell_1,\ldots,\ell_{d-1}\in\zr{D+1}\\\ell_{j-1}+1=\ell_{j}}}(X_{i_1})_{\ell_0,\ell_1}(X_{i_2})_{\ell_1,\ell_2}\cdots(X_{i_d})_{\ell_{d-1},\ell_d}\\
		&=^\dagger\begin{cases}
			x_{i_1,1}\cdots x_{i_d,d}	&	d=\ell\\
			0				&	\text{else}
		\end{cases}	\\
		&=\varphi(\homp{\ell}(x_{i_1}\cdots x_{i_d}))
	\end{align*}
	as desired, where $(\dagger)$ follows from the observation that the only way a sequence of length $d+1$ can start at $0$ and end at $\ell$, increasing 1 at each step, is if $d=\ell$, and $\ell_j=j$.
\end{proof}

We now use this to give the black-box reduction from non-commutative polynomials to set-multilinear polynomials.

\begin{lemma}
	\label{lemma:nc-to-sm-nz}
	Assume the setup of \autoref{construction:nc-to-sm}.  Let $f\in\F\{x_i\}_{i\in\zr{n}}$ be a non-commutative polynomial of degree $\le D$.  Then for each $\ell\in\zr{D+1}$, the (commutative) polynomial $(f(\vec{X}_D))_{0,\ell}$ is set-multilinear, and $f\equiv 0$ iff for all $\ell\in\zr{D+1}$, $(f(\vec{X}_D))_{0,\ell}\equiv 0$.  Consequently, $f\equiv 0$ iff $f(\vec{X}_D)\equiv 0$.
\end{lemma}
\begin{proof}
	As $\deg(f)\le D$, it is clear that $f\equiv 0$ iff for all $\ell\in\zr{D+1}$, $\homp{\ell}(f)\equiv 0$.  Further, we see that $\varphi$ is a bijection, so $\homp{\ell}(f)\equiv 0$ iff $\varphi(\homp{\ell}(f))\equiv 0$, and that $\varphi$ applied to any homogeneous non-commutative polynomial yields a set-multilinear polynomial.  Finally, we use \autoref{lemma:staircase} to see that $\varphi(\homp{\ell}(f))\equiv 0$ iff $(f(\vec{X}_D))_{0,\ell}\equiv 0$.  It follows then that $f(\vec{x})\equiv 0$ iff $f(\vec{X}_D)\equiv 0$.
\end{proof}

Thus, this lemma says that to do black-box PIT of non-commutative polynomials, it is enough to do black-box PIT of set-multilinear polynomials, in order to hit each of the $(f(\vec{X}_D))_{0,\ell}$.  However, in order for this reduction to be useful, we must show that $(f(\vec{X}_D))_{0,\ell}$ is not only set-multilinear, but is also computed by a small computational device, as otherwise no small hitting set may exist.

In the most general model, that of non-commutative circuits, it is not clear there is a simple bound on the complexity of $(f(\vec{X}_D))_{0,\ell}$ in terms of the complexity of $f$, as circuits have no global ``order'' in the computation.  That is, in a non-commutative circuit each multiplication is ordered, and this local order is fixed by the circuit.  However, as parts of the circuit may be reused, a partial computation in the circuit can appear in many different parts of the final computation.  Concretely, in the circuit of the repeated squaring computation $((x)^2)^2$, we cannot assign ``$x\to x_i$'' in any way compatible with the fact that $\varphi(x^4)= x_1x_2x_3x_4$.  Thus, it is not clear how to convert a non-commutative circuit for $f$ into a circuit for $\varphi(\homp{\ell}(f))$.

However, for weaker models of computation such as ABPs, there is such an notion of ordering as each multiplication proceeds from the source to the sink, in a single direction.  This allows us to relate the ABP complexity of a (non-commutative) polynomial $f$ to the ABP complexity of the (commutative) set-multilinear polynomial $\varphi(\homp{\ell}(f))$.  To do so, we first cite the following lemma, implicit in Nisan~\cite{Nisan91} (see also Lemma 2\ in Raz-Shpilka~\cite{RazShpilka05}\footnote{In fact, the version that we use is slightly stronger than the version in \cite{RazShpilka05}, but the proof is the same and the only difference is that we use a less wasteful analysis, which is obvious from the proof there.}).

\begin{lemma}[Nisan, \cite{Nisan91}]
	\label{lem: nonhom ABP to hom ABP}
	Let $f\in\F\{\vec{x}\}$ be a non-commutative polynomial computed by a non-commutative ABP of depth $D$ and width $r$.  Then for\footnote{The case when $\ell=0$ does not fit into this technicalities of this lemma, so will be treated separately in what follows.} $1\le\ell\le D$, $\homp{\ell}(f)$ is\footnote{In fact, these ABPs can efficiently computed given the original ABP, but we do not use this fact.} computed by a non-commutative ABP of depth $\ell$ and width $\le r(D+1)$.  Further, each edge in this ABP is labeled with a homogeneous linear form.
\end{lemma}

We now show that for ABPs as the above lemma outputs (those with only homogeneous linear forms on the edges) there is a clear connection between the ABP complexity of $f$ and $\varphi(f)$.

\begin{lemma}
	\label{lemma:nc-to-sm-abp}
	Let $f\in\F\{\vec{x}\}$ be a non-commutative polynomial computed by an ABP $A$, whose edges are labelled with homogeneous linear forms.  Let $A'$ be the ABP with the same structure as $A$, but for each edge from layer $j-1$ to layer $j$ with label $\sum_{i\in\zr{n}} a_ix_i$, we replace the label with $\sum_{i\in\zr{n}} a_ix_{i,j}$.  Then $A'$ computes $\varphi(f)$, and $A'$ is a set-multilinear ABP.
\end{lemma}
\begin{proof}
	By linearity, it is enough to prove the claim for a single path in the ABP $A$.  The path computes the product of its labels, and thus yields the product of linear forms
	\[g=\prod_j\sum_{i\in\zr{n}} a_{i,j}x_i
		=\sum_{i_k\in\zr{n}}(\prod_{j} a_{i_j,j})(\prod_jx_{i_j})
		\;.
	\]
	Similarly, the same path in $A'$ computes
	\[\prod_j\sum_{i\in\zr{n}} a_{i,j}x_{i,j}
		=\sum_{i_k\in\zr{n}}(\prod_{j} a_{i_j,j})(\prod_jx_{i_j,j})
	\]
	which equals $\varphi(g)$ by linearity, as desired. Finally, we see that the edges from layer $j-1$ to layer $j$ only involve the variables $x_{\bullet,j}$, so the ABP is indeed set-multilinear.
\end{proof}

Combining the above results, we can now give the full reduction.

\begin{theorem}
	\label{thm:ncABPs}
	Assume the setup of \autoref{construction:nc-to-sm}. Let $\cH$ be a hitting set for set-multilinear (commutative) polynomials with the (known) variable partition $\sqcup_{j\in\zr{D}} \{x_{i,j}\}_{i\in\zr{n}}$ computed by set-multilinear ABPs of width $\le r(D+1)$ and depth $\le D$. Define $\cH'\subseteq(\F^{\zr{D+1}\times\zr{D+1}})^{\zr{n}}$ by replacing each evaluation point in $\cH$ with the corresponding evaluation point defined by the matrices $\vec{X}_D$.  Then $\cH$ is a hitting set for non-commutative ABPs of width $r$ and depth $D$.
\end{theorem}
\begin{proof}
	Let $f$ be computed by a non-commutative ABPs of width $r$ and depth $D$.  Then, by \autoref{lemma:nc-to-sm-nz}, $f\equiv 0$ iff for all $\ell\in\zr{D+1}$, the set-multilinear polynomial $(f(\vec{X}_D))_{0,\ell}\equiv 0$.  If $\ell=0$, then we can observe that as the $X_i$ are strictly upper-triangular, it happens that $(f(\vec{X}_D))_{0,0}$ is constant for any setting of $\vec{X}_D$ (and in fact equals the constant term of $f$), so it is then clear that $(f(\vec{X}_D))_{0,0}\equiv 0$ iff $(f(\vec{X}_D))_{0,0}|_{\vec{X}_D\in\cH'}\equiv 0$ as $|\cH'|>0$.

	If $\ell>0$ then we use \autoref{lemma:staircase} and \autoref{lemma:nc-to-sm-nz} to see that $(f(\vec{X}_D))_{0,\ell}=\varphi(\homp{\ell}(f))$ is a set-multilinear polynomial, and \autoref{lem: nonhom ABP to hom ABP} and \autoref{lemma:nc-to-sm-abp} imply that $\varphi(\homp{\ell}(f))$ is computed by a width $\le rD$, depth $\le D$ set-multilinear ABP, and thus $(f(\vec{X}_D))_{0,\ell}\equiv 0$ iff $(f(\vec{X}_D))_{0,\ell}|_{\vec{X}_D\in\cH'}\equiv 0$.
\end{proof}

Plugging in our hitting set for set-multilinear ABPs from \autoref{thm:main:sm}, we obtain the following corollary.

\begin{corollary}[PIT for non-commutative ABPs]\label{cor:ncABP}
	Let $\mathcal{NC}$ be the set of $n$-variate non-commutative polynomials computable by width $r$, depth $D$ ABPs. If $|\F|\ge (2(D+1)^5nr^4)^2$, then $\mathcal{NC}$ has a $\poly(D,n,r)$-explicit hitting set over $(D+1)\times(D+1)$ matrices, of size $\le (2(D+1)^4n^2r^3)^{\lceil \lg D\rceil +1}$.
\end{corollary}

\begin{remark}
	The use of $(D+1)\times(D+1)$ strictly upper-triangular matrices in the above results is optimal, in the sense that if we used $m\times m$ strictly upper-triangular matrices for $m<D+1$ then no such result is possible, as any such matrix $X$ has $X^D=0$, and the monomial $x^D$ is computable by a depth $D$ ABP.
\end{remark}


\section{Diagonal Circuits}\label{sec:diagonal}


The model of diagonal circuits was first defined by Saxena~\cite{Saxena08} in order to better understand depth-4 circuits.   Much previous PIT research focused on small-depth circuits with restricted top fan-in.  In contrast, the diagonal model allows unbounded top fan-in, but each multiplication gate only has few distinct factors (but the factors can have high multiplicity).

To ease reading we shall make use of the following shortened notation.
Given a vector $\vec{e}\in\N^n$, we denote $|\vec{e}|_1\eqdef e_1+\cdots+e_n$, $|\vec{e}|_\infty=\max_i e_i$, $|\vec{e}|_\times=\prod_{i\in[n]}(1+e_i)$, and $\vec{e}\,!=e_1!\cdots e_n!$.

For depth-3 diagonal circuits, each multiplication gate has the form $\vec{\ell}(\vec{x})^{\vec{e}}$, for affine linear forms $\ell_i$, where
$|\vec{e}\:\!|_\times$ is bounded by $\poly(n)$.  Saxena's techniques also extend to the case when the $\ell_i$ are sum of univariate polynomials, which is the depth-4 diagonal circuit model. We now give a formal definition.

\begin{definition}[Diagonal Circuits]\label{def: diagonal}
	A \textbf{diagonal depth-4 circuit} has the form $\Phi = \sum_{i=1}^{k}\Psi_i$, where each product gate $\Psi_i$ is of the form $\Psi_i = \vec{P}_i^{\vec{e}_i}$, and each $P_{i,j}(\vec{x})$ is a sum of univariate polynomials, $P_{i,j} = \sum_{m=1}^{n} g_{i,j,m}(x_m)$. The syntactic degree of the polynomial computed by the circuit is $\sdeg(\Phi) = \max_i\deg(\Psi_i)$, where $\deg(\Psi_i) = \sum_j e_{i,j}\deg(P_{i,j})$ (i.e. the syntactic degree of $\Phi$ is the maximal degree of a multiplication gate $\Psi_i$ in $\Phi$).
\end{definition}

Saxena~\cite{Saxena08} proved the following theorem, establishing a white-box PIT algorithm for diagonal depth-4 circuits.

\begin{theorem}[Saxena, \cite{Saxena08}]
	\label{thm: saxena diagonal}
	There is a deterministic algorithm that, given as input a diagonal depth-4 circuit $\Phi$, runs in time $\poly(nk\cdot\sdeg(\Phi),\max_{i \in [k]} {|\vec{e}_i|_\times})$ and decides whether $\Phi \equiv 0$.
\end{theorem}

We give a black-box PIT algorithm for diagonal depth-4 circuits whose running time is quasi-polynomial in the runtime established by Saxena. Saxena gave a white-box reduction from diagonal depth-4 circuits to non-commutative PIT over algebras, and claimed that the Raz-Shpilka~\cite{RazShpilka05} algorithm can be extended to handle this case.  Our result uses a poly-time black-box reduction from diagonal depth-4 circuits to read-once oblivious ABPs.  We follow Saxena's duality idea, but ``pull'' the complexity of the algebras he works with into the ABP itself.  We can then apply our hitting sets on read-once oblivious ABPs, which incurs the quasi-polynomial overhead.

Our reduction will be based on the duality idea of Saxena~\cite{Saxena08} (see his Lemma 2). However, the prior duality statements were cumbersome over fields of positive characteristic.  Saxena~\cite{Saxena08} only alluded briefly to this case, and later Saha, Saptharishi and Saxena~\cite{SahaSS11} gave a description, arguing that duality holds if one can move from $\Z_p$ to $\Z_{p^N}$ for some large $N$.  This is not suitable for black-box reductions, as we cannot change the characteristic in a black-box way.  Thus, we develop a unified duality statement that works for all fields, and also works in the black-box setting.  This duality improves on that of Saxena~\cite{Saxena08} by examining \textit{individual degree} as opposed to \textit{total degree}, and using this technical improvement we can extend the result to positive characteristic by using the Frobenius automorphisms.

We will denote the exponential power series as $\exp(x)=\sum_{\ell\ge 0} \frac{x^\ell}{\ell !}$, and the truncated exponential by $\trexp_e(x) = \sum_{\ell=0}^{e} \frac{x^\ell}{\ell !}$.
Further, recall our meaning, stated in \autoref{sec:notation}, for extracting coefficients of polynomials, and its meaning when we take coefficients ``over $\F[\vec{x}][\vec{y}]$'' (as opposed to `` over $\F[\vec{x},\vec{y}]$'').

We first establish duality over $\Q$, treating the coefficients $\vec{c}$ of our polynomials as formal variables.

\begin{lemma}[Duality over $\Q$]
	\label{lem:dualQ}
	Let $\Psi=\vec{P}^{\vec{e}}\in\Z[\vec{x},\vec{c}]$, where $P_j(\vec{x},\vec{c})=\sum_m g_{j,m}(x_m,\vec{c})$. Then, taking coefficients in $\Q[\vec{x},\vec{c}][\vec{z}]$,
	\begin{equation}
		\label{eq:dual}
		\frac{1}{\vec{e}\,!}\Psi
		=
		\coeff{\vec{z}^{\vec{e}}}
		\left(\prod_m\prod_j\trexp_{e_j}\left(g_{j,m}(x_m,\vec{c}) z_j \right)\right).
	\end{equation}
\end{lemma}
\begin{proof}
	We will work over the power series $\Q\llb \vec{x},\vec{c},\vec{z}\rrb$.  In particular, for $f\in\Q[\vec{x},\vec{c}]$, we have the identity \[\frac{1}{e!}f(\vec{x},\vec{c})^{e}=\coeff{z^{e}}\exp(f(\vec{x},\vec{c})z)\;.\]  Note that this is a formal identity, that is, there is no notion of convergence needed as each coefficient in the power series of $\exp(f(\vec{x},\vec{c})z)$ is the sum of finitely many non-zero terms.  It follows then that
	\begin{align*}
		\frac{1}{\vec{e}\,!}\Psi
		&=\prod_j \coeff{z_j^{e_j}} \exp(P_j(\vec{x},\vec{c})z_j)
		=\coeff{\vec{z}^{\vec{e}}}\prod_j \exp(P_j(\vec{x},\vec{c})z_j)\\
		\intertext{using the identity $\exp(f+g)=\exp(f)\exp(g)$, which formally holds in $\Q\llb \vec{x},\vec{c},\vec{z}\rrb$,}
		&=\coeff{\vec{z}^{\vec{e}}}\exp\left(\sum_j P_j(\vec{x},\vec{c})z_j\right)
		=\coeff{\vec{z}^{\vec{e}}}\exp\left(\sum_j \sum_m g_{j,m}(x_m,\vec{c})z_j\right)\\
		\intertext{reversing the order of summation,}
		&=\coeff{\vec{z}^{\vec{e}}}\exp\left(\sum_m \sum_j g_{j,m}(x_m,\vec{c})z_j\right)\\
		&=\coeff{\vec{z}^{\vec{e}}}\prod_m \prod_j\exp\left(g_{j,m}(x_m,\vec{c})z_j\right)\\
		\intertext{and now discarding terms of degree $>e_j$ in $z_j$, as that do not affect the coefficient of $\vec{z}^{\vec{e}}$,}
		&=\coeff{\vec{z}^{\vec{e}}}\prod_m \prod_j\trexp_{e_j}\left(g_{j,m}(x_m,\vec{c})z_j\right)\qedhere
	\end{align*}
\end{proof}

Even though the above proof used the exponential function, which has no exact analogue in positive characteristic, the final statement is just one of polynomials, and thus transfers to any field of sufficiently large characteristic.  In particular, the characteristic must be $>|\vec{e}\:\!|_\infty$, the maximum exponent. We now transfer this duality to any field. We will abuse notation and treat the characteristic zero case as working over characteristic $p$, where $p=\infty$.  Note that the ``base-$\infty$'' expansion of any integer is just that integer itself, followed by zeroes.

\begin{lemma}[Duality over any $\F$]
	\label{lem: dual of diagonal}
	Let $\F$ be a field of characteristic $p$ (for characteristic zero, we abuse notation and use $p=\infty$).  Let $\Psi=\vec{P}^{\vec{e}}\in\F[\vec{x}]$, where $P_j(\vec{x})=\sum_m g_{j,m}(x_m)$.

	Write $e_j$ in its base-$p$ expansion, so that $e_j=\sum_{k\in\zr{n_j}} a_{j,k}p^k$ with $a_{j,k}\in\zr{p}$.  Denote $Q_{j,k}(\vec{x})\eqdef\sum_m g_{j,m}(x_m)^{p^k}$ and $\vec{Q}^{\vec{a}}\eqdef\prod_j\prod_k Q_{j,k}^{a_{j,k}}$.  Then, $\vec{P}^{\vec{e}}=\vec{Q}^{\vec{a}}$ and
	\begin{equation}
		\label{eq:dualp}
		\frac{1}{\vec{a}!}
		\vec{Q}^{\vec{a}}
		=
		\coeff{\vec{z}^{\vec{a}}}
		\left(\prod_m\prod_{j,k}\trexp_{a_{j,k}}\left(g_{j,m}(x_m)^{p^k} z_j \right)\right),
	\end{equation}
	and $|\vec{a}|_\times\le |\vec{e}|_\times$.
\end{lemma}
\begin{proof}
	\uline{$\vec{P}^{\vec{e}}=\vec{Q}^{\vec{a}}$:} This follows from the Frobenius automorphism, that $(x+y)^{p^k}=x^{p^k}+y^{p^k}$ for any $k\ge 0$, in any field of characteristic $p$. That is,
	\begin{align*}
		\vec{P}^{\vec{e}}
		&=\prod_j P_j^{e_j}
		=\prod_jP_j^{\sum_{k\in\zr{n_j}}a_{j,k}p^k}	
		=\prod_j\prod_{k} P_j^{a_{j,k}p^k}
		=\prod_j\prod_k (P_j^{p^k})^{a_{j,k}}\\
		&=\prod_j\prod_k \left(\left(\sum_m g_{j,m}(x_m)\right)^{p^k}\right)^{a_{j,k}}\\
		\intertext{by the Frobenius automorphism,}
		&=\prod_j\prod_k \left(\sum_m g_{j,m}(x_m)^{p^k}\right)^{a_{j,k}}
		=\prod_j\prod_k Q_{j,k}^{a_{j,k}}
		=\vec{Q}^{\vec{a}}\;.
	\end{align*}

	\uline{\eqref{eq:dualp}:} Note that the identity of \autoref{lem:dualQ} only involves rational numbers whose denominators $\le |\vec{e}|_\infty$, the maximum exponent.  Thus, as $|\vec{a}|_\infty<p$ by construction, \autoref{lem:dualQ} is also an identity over $\Z_p[\vec{x},\vec{c}]$ when applied with the exponent vector $\vec{a}$, and any $g$ polynomials.  In particular, we can substitute values for the constants $\vec{c}$, implying that we can take the $g$ polynomials to be polynomials in $\vec{x}$ with constants from any field extending $\Z_p$, such as $\F$. Thus, the identity can be applied to $\vec{Q}^{\vec{a}}$ over $\F$, yielding the result.

	\uline{$|\vec{a}|_\times\le |\vec{e}\:\!|_\times$:} As $|\vec{e}\:\!|_\times=\prod_j(1+e_j)$ and $|\vec{a}|_\times=\prod_j\prod_{k\in\zr{n_j}}(1+a_{j,k})$, it is sufficient to prove that $\prod_k (1+a_k)\le 1+e$, for $e=\sum_k a_k p^k$ with $a_k\in\zr{p}$.  To see this, first define $T\eqdef\{e':0\le e'\le e\}$ and note that $|T|=1+e$.  Second, let $S=\{e': e'=\sum_k a'_k p^k, 0\le a'_k\le a_k\}$.  Observe that $S\subseteq T$, and $|S|=\prod_k (1+a_k)$, so $\prod_k (1+a_k)=|S|\le |T|=1+e$.
\end{proof}


To use this lemma, we first need a structural result about ABPs.  In this lemma we broaden the notion of an ABP to have arbitrary labels from a commutative ring $\mathcal{R}[\vec{x}]$, and will specialize this to our setting later.

\begin{lemma}
	\label{lemma: homogenize}
	Let $A$ be an layered ABP of depth $D$ and width $\le r$, with edge labels in $\cR[\vec{z}]$ for some commutative ring $\cR$, computing an $n$-variate polynomial $f\in \cR[\vec{z}]$.  Then for any $\vec{e}\in\N^n$, there is an ABP $A'$, with edge labels in $\cR$, of depth $D$ and width $\le r\cdot |\vec{e}|_\times$ computing $\coeff{\vec{z}^{\vec{e}}}(f)$.
		
	Further, the edge labels between layers $i-1$ and $i$ in $A'$ occur as coefficients of the edge labels between layers $i-1$ and $i$ in $A$.
\end{lemma}
\begin{proof}
	Define $S\eqdef \{\vec{a}: \vec{0}\le\vec{a}\le \vec{e}, \vec{a}\in\N^n\}$, where we impose the natural partial order on $\N^n$. Clearly $|S|=|\vec{e}|_\times$. Index the nodes in $A$ by $[r]\times\zr{D+1}$, so that the $i$-th layer in $A$ consists of the nodes $[r]\times\{i\}$. Define the nodes of $A'$ to be $[r]\times S\times\zr{D+1}$, so that the $i$-th layer of $A'$ consists of the nodes $[r]\times S\times\{i\}$.

	For each edge $(v\times(i-1),v'\times i)$ in $A$ with label $g\in \cR[\vec{z}]$, and each $\vec{a},\vec{a}'\in S$ with $\vec{a}\le \vec{a}'$, define the edge label $(v\times\vec{a}\times(i-1),v'\times\vec{a}'\times i)$ to be $\coeff{\vec{z}^{\vec{a}'-\vec{a}}}(g)$. Letting the source of $A$ be denoted $s\times 0$, and the sink denotes $t\times D$, we define the source of $A'$ to be $s\times\vec{0}\times 0$ and the sink to be $t\times\vec{e}\times D$.  After removing nodes that belong to no source-sink path, it is clear that $A'$ is a layered ABP, with depth $D$, width $\le r\cdot |\vec{e}|_\times$.  Further, the edge labels between layers $i-1$ and $i$ in $A'$ are coefficients of the edge labels between layers $i-1$ and $i$ in $A$ as desired.  It remains to show that $A'$ computes as desired, which follows from an induction on layers.  Specifically, one can see that the paths from the source of $A'$ to $v\times\vec{a}\times i$ compute $\coeff{\vec{z}^{\vec{a}}}(f_{v\times i})$, where $f_{v\times i}$ is the polynomial computed by the paths from the source in $A$ to $v\times i$.
\end{proof}

We now apply this structural result, along with Saxena's dual form, to show that depth-4 diagonal circuits can be computed by small read-once oblivious ABPs.

\begin{lemma}\label{lem: diagonal to ABP}
	Let $\F$ be any field. Let $f(x_1,\ldots,x_n)$ be computed by a diagonal circuit $\Phi = \sum_{i=1}^{k}\Psi_i$, where $\Psi_i = \vec{P}_i^{\vec{e}_i}$. Denote the syntactic degree of $f$ to be $\sdeg(f)=\max_i \sdeg(\Psi_i)$. Then there is a read-once oblivious ABP $A$ computing $f$ with variable order $x_1<\cdots<x_n$, with depth $n$ and width $\le k \cdot\max_{i\in[k]}|\vec{e}_i|_\times$ such that the edges in $A$ are labeled with polynomials with degrees $\le \sdeg(f)$.
\end{lemma}
\begin{proof}
	Let $p$ denote the characteristic of $\F$, with the convention that $p=\infty$ for characteristic zero fields. By the dual form, \autoref{lem: dual of diagonal}, we have that for each $i$,
	\[\Psi_i=\vec{a}_i!\cdot \coeff{\vec{z}^{\vec{a}_i}}(f_i)\]
	where $\vec{a}_i$ is the base-$p$ expansion of $\vec{e}_i$ (and $|\vec{a}_i|_\times\le|\vec{e}_i|_\times$), $f_i\in\F[\vec{x}][\vec{z}]$, and $f_i$ is a product of $n$ terms, where the $j$-th term only involves $x_j$ amongst the $\vec{x}$ variables.  It follows that $f_i$ can be computed using a width-1, depth-$n$ ABP in $\F[\vec{x}][\vec{z}]$, where the labels from layer $j-1$ to $j$ only involve $x_j$ amongst all $\vec{x}$ variables.
	
	Applying \autoref{lemma: homogenize} to $f_i$, taking $\mathcal{R}=\F[\vec{x}]$, we see (after a scalar multiplication by $\vec{a}_i!$) that $\Psi_i$ is computable by an ABP $A_i$, which depth $n$, and width $\le |\vec{a}_i|_\times\le |\vec{e}_i|_\times$.  Further, \autoref{lemma: homogenize} establishes that the labels in $A_i$ from layer $j-1$ to layer $j$ come from the labels (by taking coefficients) on the edges from layer $j-1$ to layer $j$ in the ABP computing the product $f_i$ in $\F[\vec{x}][\vec{z}]$.  Thus, we see that the labels from layer $j-1$ to layer $j$ in $A_i$ are polynomials in $x_j$, and have degree at most $\sdeg(\Psi_i)$, and thus $A_i$ is a read-once oblivious ABP.

	Now, observe that the read-once oblivious ABPs for the $\Psi_i$ can be summed by merging all sources and merging all sinks, and that the resulting ABP $A$ is read-once oblivious, as the variable ordering of $\vec{x}$ is consistent amongst the $A_i$.  It follows that $A$ has the desired properties.
\end{proof}

As the above lemma is constructive, it follows that white-box access to $f$ implies white-box access to the read-once oblivious ABP $A$, and thus we could run the white-box algorithm of Raz-Shpilka~\cite{RazShpilka05} to derive a white-box PIT algorithm, in order to rederive Saxena's result (over any field).  However, as we also have black-box PIT algorithms for read-once oblivious ABPs (and the above uses the fixed variable order $x_1<\cdots<x_n$), and the above reduction does not need access to $f$, we can combine these results to deduce the following black-box PIT result for diagonal depth-4 circuits.

\begin{theorem}[Black-box PIT for diagonal circuits]\label{thm:diagonal}	
	Let $\F$ be a field of arbitrary characteristic.  Let $\mathcal{DC}$ be the set of $n$-variate polynomials computable by diagonal depth-4 circuits, that is, of the form $\Phi=\sum_{i\in[k]}\vec{P}_i^{\vec{e}_i}$, where $|\vec{e}_i|_\times\le e$ and $P_{i,j}$ is of degree $\le d$.  Then if $|\F|\ge (2ndk^3e^4)^2$ then $\mathcal{DC}$ has a $\poly(n,k,e,d)$-explicit hitting set of size $\le (2nd^2k^4e^6)^{\lceil \lg n\rceil +1}$.
\end{theorem}
\begin{proof}
	From \autoref{lem: diagonal to ABP} we get that any $\Phi\in\mathcal{DC}$ can be computed by a read-once oblivious ABP on variable order $x_1<\cdots<x_n$, of depth $n$, width $\le ke$, with each edge label being a univariate of degree $\le\sdeg(\Psi)$.  Noting that $\sdeg(\vec{P}_i^{\vec{e}_i})\le |\vec{e}_i|_1\cdot d< |\vec{e}_i|_\times\cdot d$, we invoke \autoref{thm:main} to finish the claim.
\end{proof}

\begin{remark}
	We remark here that the concurrent work of Agrawal, Saha and Saxena~\cite{AgrawalSS12} also obtain a quasi-polynomial hitting set for diagonal depth-4 circuits (when they assume each product gate is over a constant number of factors), but only over large characteristic, as they rely on the duality statements of Saxena~\cite{Saxena08} (and later exposited by Saha, Saptharishi, and Saxena~\cite{SahaSS11}) which only hold over large characteristic. Over small characteristic, \cite{Saxena08} and \cite{SahaSS11} gave more cumbersome duality statements working over prime-power characteristic, and \cite{AgrawalSS12} do not extend their work to this case.

	Our duality statements work any characteristic, and as such can show that the work of \cite{AgrawalSS12} also implies results for diagonal circuits over any characteristic.  In particular, instead of using \autoref{lem: dual of diagonal} to construct an ABP, one can interpolate the coefficient $\vec{z}^{\vec{e}}$ in \autoref{lem: dual of diagonal}, and this can be done in depth-3, although the resulting formula is inherently larger than the resulting ABP would be.
\end{remark}

\subsection{Semi-diagonal Depth-4 circuits}

In Saha, Saptharishi, and Saxena~\cite{SahaSS11} the model of semi-diagonal depth-4 circuits was introduced as a small extension of diagonal depth-4 circuits. The modification is that one is allowed to multiply each product gate $\Psi_i=\vec{P}^{\vec{e}}$ by an arbitrary monomial. \cite{SahaSS11} used that the duality result of Saxena~\cite{Saxena08} can also be shown to work in this setting.  The concurrent work of Agrawal, Saha and Saxena~\cite{AgrawalSS12} thus present their results for semi-diagonal depth-4 circuits, as opposed to just diagonal depth-4 circuits.  For ease of comparison, we also state our result for this model.  As multiplication by a single monomial in \autoref{lem: dual of diagonal} preserves the variable disjoint product structure (but increases the degree), it follows that we can still convert semi-diagonal circuits to read-once oblivious ABPs, as stated in \autoref{lem: diagonal to ABP}.  Thus, as the details are the same, we conclude the following result.

\begin{theorem}[Black-box PIT for semi-diagonal circuits]\label{thm: semi diagonal}	
	Let $\F$ be a field of arbitrary characteristic.  Let $\mathcal{SDC}$ be the set of $n$-variate polynomials computable by semi-diagonal depth-4 circuits, that is, of the form $\Phi=\sum_{i\in[k]}m_i(\vec{x})\vec{P}_i^{\vec{e}_i}$, where $m_i(\vec{x})$ is a monomial of degree $\le d$, $|\vec{e}_i|_\times\le e$ and $P_{i,j}$ is of degree $\le d$.  Then if $|\F|\ge (4ndk^3e^4)^2$ then $\mathcal{SDC}$ has a $\poly(n,k,e,d)$-explicit hitting set of size $\le (8nd^2k^4e^6)^{\lceil \lg n\rceil +1}$.
\end{theorem}


\section{Evaluation dimension}
\label{sec: eval-dim}

The content of this section was communicated to us by Ramprasad Saptharishi~\cite{Saptharishi12}.

\begin{definition}[Evaluation dimension]
	A polynomial $f \in \F[x_1,\ldots ,x_n]$ is said to have evaluation dimension (denoted by $\evaldim(f)$) $r$ if for any subset of variables $S = \{i_1,...,i_k\} \subseteq [n]$
	\[\dim\left(\spn\left\{f|_{x_{i_j} = \alpha_j,j\in[k]} : (\alpha_1,\ldots,\alpha_k) \in\F\right\}\right)\leq r.\]
\end{definition}

We leave the proofs of the next lemmas to the reader.

\begin{lemma}
	 If the dimension of the partial derivative space of $f$ is bounded by $r$, then $\evaldim(f) \leq r$.
\end{lemma}

Note that the converse is not necessarily true as demonstrated by the polynomial $(\sum_{i=1}^{n}x_i^2)^n$.

\begin{lemma}
	\sloppy Depth-$4$ diagonal circuits (as given in \autoref{def: diagonal}) have $\poly(nk\cdot\sdeg(\Phi),\max_{i \in [k]} {|\vec{e}_i|_\times})$ evaluation dimension.
\end{lemma}

Saptharishi~\cite{Saptharishi12} observed that our proof technique for constructing hitting sets for read-once oblivious ABPs can be applied also to polynomials with small evaluation dimension.  We give here an alternate proof of that fact, showing that any polynomial with small evaluation dimension can be computed by a small width read-once oblivious ABP.

\begin{theorem}
	Let $f$ be an $D$-variate, degree $<n$ polynomial, of evaluation dimension $\le r$.  Then $f$ can be computed by a width $\le rn^2$, depth $D$, degree $<n$ read-once oblivious ABP (in any variable ordering).
\end{theorem}
\begin{proof}[Sketch of Proof]
	One can note, by polynomial interpolation, that for any set $S\subseteq[n]$, the dimension
	\[\dim\left(\spn\left\{f|_{x_{i_j} = \alpha_j,j\in[k]} : (\alpha_1,\ldots,\alpha_k) \in\F\right\}\right)\] is equal to the rank of the \textit{partial derivative matrix} of the variable partition $[n]=S\sqcup \bar{S}$, as defined by Nisan~\cite{Nisan91} in the context of non-commutative computation.  Now fix an arbitrary variable ordering.  As read-once oblivious ABPs invoke variables in this fixed ordering, they can be seen as non-commutative ABPs with all monomials respecting the variable ordering.  Nisan showed that for a homogeneous non-commutative polynomial $f$, one can construct a non-commutative ABP computing $f$ whose width is equal to the maximum rank of a partial derivative matrix of $f$, where the partition $S\sqcup \bar{S}$ respects the non-commutative multiplication. Applying this result to read-once oblivious ABPs we observe that, by standard homogenization, each homogeneous part of the target polynomial $f$ has evaluation dimension at most $rn$, so can be computed by a width $rn$ read-one oblivious ABP.  Summing up the $n$ homogeneous parts gives the result.
\end{proof}

Applying this result with \autoref{thm:main} we get the following corollary.

\begin{corollary}[Saptharishi~\cite{Saptharishi12}]
	The class of $n$-variate degree-$d$ polynomials of evaluation dimension bounded by $r$ has a black-box PIT running in time $\poly(n, d, r)^{\O(\lg n)}$.
\end{corollary}

Finally, we mention that Saptharishi~\cite{Saptharishi12} also showed that evaluation dimension can be used in the white box setting to obtain new PIT algorithms for semi-diagonal depth-$4$ circuits.

\section{Discussion}

This work closes some gap in our understanding of white-box PIT vs.\ black-box PIT by transferring algorithms that used the partial derivative technique to the black-box world. The recent work of \cite{AgrawalSS12} made another significant step by considering set-multilinear formulas (of small depth) where the partition is unknown. It will be very interesting to try and combine these two works to obtain a PIT algorithm for  set-multilinear ABPs when the underlying partition is not known.

Another interesting goal is to truly close the gap between white-box and black-box. Specifically, all black-box algorithms for the models studied in this paper (as well as in \cite{AgrawalSS12}) run with a quasi-polynomial overhead over the run-times of the corresponding known white-box algorithms. Ideally, this overhead could be made polynomial.

Finally, it will be interesting to understand whether analog of \autoref{thm:small width} can be obtained in the Boolean setting using our ideas.

\section{Acknowledgments}

Much of this work was done when the first author was visiting the second author at the Technion, some while the first author was an intern at Microsoft Research Silicon Valley.  We would like to thank Andy Drucker, Omer Reingold, Ramprasad Saptharishi, Ilya Volkovich and Sergey Yekhanin for some helpful conversations. We thank Ketan Mulmuley for sharing \cite{Mulmuley12} with us. We thank Ramprasad for allowing us to include his result on evaluation dimension (see \autoref{sec: eval-dim}) here. We also thank Avi Wigderson for raising the question of whether our technique could yield better results for the bounded width case.

\bibliographystyle{alphaurl}
\bibliography{bibliography}

\end{document}